\documentclass[12pt, draftclsnofoot, onecolumn]{IEEEtran}
\usepackage{amsmath,epsfig,epsf,times,amssymb,booktabs,subfigure}
\usepackage{setspace}
\usepackage{array}
\usepackage{url}
\usepackage{float}
\usepackage{color}
\usepackage{graphicx}
\usepackage{amsthm,cite}
\usepackage{epstopdf}
\usepackage{enumerate}

\graphicspath{{fig/}}

\newtheorem{theorem}{Theorem}

\newtheorem{lemma}{Lemma}

\theoremstyle{remark}

\newcommand{\lb}{\left(}
\newcommand{\rb}{\right)}
\newcommand{\lc}{\left\{}
\newcommand{\rc}{\right\}}

\newcommand{\abs}[1]{ \left| #1 \right|}

\begin{document}

\title{Millimeter-Wave Beam Search with Iterative Deactivation and Beam Shifting}
\author{Chunshan Liu, Min Li, Lou Zhao, Philip Whiting, Stephen~V.~Hanly, Iain~B.~Collings
\thanks{This work will be presented in part in Proceedings of the 54th IEEE International Conference on Communications, ICC'20, Dublin, Ireland, June 7-11, 2020~\cite{Liu2020ICC}. Chunshan Liu and Lou Zhao are with the School of Communication Engineering, Hangzhou Dianzi University, Hangzhou, 310018, China (email: \{chunshan.liu, lou.zhao\}@hdu.edu.cn). Min Li is with College of Information Science and Electronic Engineering, Zhejiang University, Hangzhou, 310027, China (e-mail: min.li@zju.edu.cn). Philip Whiting, Stephen V. Hanly and Iain B. Collings are with the School of Engineering, Macquarie University, Sydney, Australia (email: \{philip.whiting, stephen.hanly, iain.collings\}@mq.edu.au). (\textit{Corresponding author: Min Li})}}

\maketitle
\begin{abstract}
Millimeter Wave (mmWave) communications rely on highly directional beams to combat severe propagation loss. In this paper, an adaptive beam search algorithm based on spatial scanning, called Iterative Deactivation and Beam Shifting (IDBS), is proposed for mmWave beam alignment. IDBS does not require advance information such as the Signal-to-Noise Ratio (SNR) and channel statistics, and matches the training overhead to the unknown SNR to achieve satisfactory performance. The algorithm works by gradually deactivating beams using a Bayesian probability criterion based on a uniform improper prior, where beam deactivation can be implemented with low-complexity operations that require computing a low-degree polynomial or a search through a look-up table. Numerical results confirm that IDBS adapts to different propagation scenarios such as line-of-sight and non-line-of-sight and to different SNRs. It can achieve better tradeoffs between training overhead and beam alignment accuracy than existing non-adaptive algorithms that have fixed training overheads.
\end{abstract}

\section{Introduction}
Millimetre wave (mmWave) wireless communications represent one of the most promising means of introducing new high bandwidth services as anticipated for 5G and beyond (such as high-definition video, V2X communications~\cite{pi2011introduction,andrews2014will,xiao2017millimeter,7914742,lee2018spectrum,niu2015survey,8651496}, mobile edge computing~\cite{8959381} and surveillance~\cite{8335327}). With strongly directional beams implemented by the antenna arrays equipped at the base station (BS) and the user equipment (UE) to compensate for the high propagation loss~\cite{Akdeniz2014}, high data rate mmWave communications can be achieved over distances of hundreds of metres. However, to align the directional beams to dominant propagation paths requires beam direction searches, which are challenging~\cite{8458146}  because they often have to take place over a large search space and within a very short time period. The former is due to the requirement of highly directional beams, hence requiring a large number of beams to cover all possible directions. The latter is due to the short coherence time of mmWave channels which can be on the order of fractions of a millisecond in mobile scenarios~\cite{va2016impact}.

Spatial scanning is a widely adopted approach for mmWave beam search
\cite{hur2013millimeter,alkhateeb2014channel,tang2018high,liu2017Jsac,3GPPR1,3GPPR2,Xiao2016,zhang2017codebook,yaman2016reducing} and has drawn considerable attention due to its simplicity and high performance. In spatial scanning, pre-synthesised beams covering the angular intervals of interest are examined so as to determine the best possible BS-UE beam pair that aligns with the dominant path of the channel.\footnote{If prior knowledge of the direction of the dominant path is available, spatial scanning can be performed in a reduced space, e.g., over path skeletons~\cite{khosravi2019efficient}. In this work, we focus on initial alignment where there is no such prior knowledge.} Once this BS-UE beam pair has been found, it can be used immediately for subsequent data transmission, and thus explicit estimate of the channel coefficients is not needed. Due to the sparsity of mmWave channels, such beams, upon correct identification, can provide spectrum efficiencies very close to that from the optimal BS-UE beams constructed with perfect channel knowledge~\cite{raghavan2016beamforming}.

Much effort has been devoted to developing beam search algorithms that attain satisfactory performance (e.g., the achievable spectrum efficiency using the beams identified) with short training time~\cite{alkhateeb2014channel,Xiao2016,zhang2017codebook,liu2017Jsac,min2019TWC,tang2018high}. The ability to achieve satisfactory beam search performance with less training time is clearly a desirable goal as it gives more time for data transmission. Such  algorithms also reduce access delay caused by beam search, which is helpful in meeting low latency targets in 5G and future networks. They also extend the coverage range of mmWave BSs, as more efficient algorithms give better chances for users further away from BSs to find their optimal beam pair within a limited training time.

Hierarchical Search (HS) is a classical approach to beam alignment~\cite{hur2013millimeter,alkhateeb2014channel}. It uses tree-search in conjunction with hierarchical multi-resolution codebooks to reduce the number of beam measurements and hence search time. (The beam search method specified in IEEE 802.11ad shares a similar spirit to HS, where searches are performed in two hierarchical levels at the sector level and the beam level~\cite{yaman2016reducing}.) Subsequent investigations into HS~\cite{Xiao2016,zhang2017codebook} proposed advanced multi-resolution codebook designs.  Recent research has shown that HS incurs significant performance loss and is inferior to Exhaustive Search (ES) when the pre-beamforming signal-to-noise ratio (SNR) of the dominant channel path is low~\cite{liu2017Jsac}. In addition, more recent study \cite{min2019TWC} has proposed a generalised ES algorithm called optimised two-stage search (OTSS) to further enhance ES. During the first stage of OTSS, beam combinations unlikely to be optimal are eliminated using only a fraction of the training time,  leaving the best beam combination to be identified from a small set of candidates in the second stage. A similar algorithm is proposed in~\cite{tang2018high}. Both methods achieve the same beam search performance as ES within shorter time.

All the above algorithms are \emph{non-adaptive} as they require the amount of training time to be optimised prior to search. This is problematic because different SNR scenarios require different amounts of time to attain satisfactory search accuracy, e.g., a high SNR scenario requires less time than a low SNR scenario to reach the same  search accuracy. And the SNR necessary to obtain the optimised training time is unknown in initial beam search and varies across different users. The drawback of using a fixed overhead for all scenarios is clear: If the training time is set to ensure the performance for low-SNR users, it will lead to unnecessarily large overhead for high-SNR users; if the training time is chosen to ensure only  high-SNR users, then low-SNR users will experience poor beam search performance.

To overcome the above shortcoming, we propose a new spatial-scanning beam search algorithm called Iterative Deactivation and Beam Shifting (IDBS). Our  proposed IDBS algorithm is an \emph{adaptive} approach in the sense that it uses a suitable amount of training time in all cases to achieve satisfactory beam search performance: In low SNR, IDBS uses longer training times to achieve good beam search accuracy, while in high SNR, IDBS uses shorter training time leaving more time for data transmission and reducing access delay. IDBS achieves this goal by adopting a Bayesian criterion, which does not require prior knowledge of the SNR or of the fading statistics.

Our proposed IDBS differs from conventional non-adaptive algorithms that have a fixed training time, which limits their flexibility in achieving good balances between beam search accuracy and the overhead spent in different scenarios. IDBS also differs from other iterative algorithms with variable training time~\cite{Gehard2019ICC,kokshoorn2017millimeter}, because IDBS makes no assumptions as to the underlying channel statistics and does not require prior knowledge of the channel. As a comparison, \cite{Gehard2019ICC} uses a reference SNR value and a multiple hypothesis testing method to calculate the initial required training time, which is then updated based on channel estimates made during beam search. The authors in~\cite{kokshoorn2017millimeter} assume i.i.d. complex Gaussian statistics for the path gains with known variance, from which the successful beam alignment probability was derived and an adaptive algorithm was developed.

IDBS owes its adaptivity to the use of a Bayesian criterion together with an improper prior. As we will show, this leads to a low complexity, robust and transparent algorithm. The best beams are thus identified using a probabilistic criterion which rejects beams unlikely to be the best when compared with the likelihood score for the current maximum. Only a single parameter in the form of an acceptance~probability is needed. As our empirical results will show, the same parameter choice gives satisfactory performance across a wide range of channels.

Another important feature of IDBS is that it can achieve higher spatial resolution than the original codebook used for training via a procedure we call beam shifting. This feature does not require any additional training overhead, including the use of an oversampled codebook for spatial scanning, a method suggested by~\cite{Xiao2016}. In fact, by allowing beam shifting, IDBS makes an overhead saving by avoiding the excessive overhead needed to distinguish between two comparable beams.

The remainder of the paper is as follows. Section \ref{Sec:system_signal} presents the problem formulation and the system model which is followed by a detailed description
of IDBS in Section \ref{Sec:algorithm}. Extensive numerical results are presented in Section  \ref{Sec:Toy_example} and Section \ref{Sec:Numerical} which are used
to examine the algorithm features and to show its superior performance as compared to non-adaptive approaches. Conclusions are drawns in Section~\ref{Sec:Conclusions}. Some theoretical results are deferred to the Appendices.

\section{System Model and Problem Formulation}\label{Sec:system_signal}
We consider a mmWave communications system where Uniform Linear Arrays (ULAs) are equipped at both BS and UE. In this system, BS and UE cooperatively send and measure pilot signals with different narrow beams pre-synthesised to jointly cover the angular intervals of interest at BS and UE. By searching through the pre-synthesised codebooks, the goal is to find the BS-UE beamformers that align well to the strongest channel path, i.e., maximising  the effective channel gain after beamforming.

Let $N_T$ be the number of antennas at BS and $N_R$ the number of antennas at UE. Denote ${\cal L} = \{\mathbf{u}_1,\ldots,\mathbf{u}_{L}\}$ as the UE codebook and ${\cal S} = \{\mathbf{w}_1,\ldots,\mathbf{w}_{S}\}$ as the BS codebook. In this work, we consider DFT beamformers in ${\cal L}$ and ${\cal S}$:
\begin{equation}\label{eq:Rx_DFT}
\mathbf{u}_l = \frac{1}{\sqrt{N_R}}\left[1,e^{-i\frac{2\pi d}{\lambda}\psi^c_l},\ldots,e^{-i\frac{2\pi d}{\lambda}(N_R-1){\psi}^c_l}\right]^T,
\end{equation}
\begin{equation}\label{eq:Tx_DFT}
\mathbf{w}_s = \frac{1}{\sqrt{N_T}}\left[1,e^{-i\frac{2\pi d}{\lambda}\phi^c_s},\ldots,e^{-i\frac{2\pi d}{\lambda}(N_T-1){\phi}^c_s}\right]^T.
\end{equation}
Here ${\psi}^c_l \doteq \sin(\tilde{\psi}^c_l )\in [-1,1]$ (${\phi}^c_s \doteq \sin(\tilde{\phi}^c_s )\in [-1,1]$) and  $\tilde{\psi}^c_l \in [-90^\circ,+90^\circ]$ ($\tilde{\phi}^c_s \in [-90^\circ,+90^\circ]$) is the beam centre of $\mathbf{u}_l$ ($\mathbf{w}_s$). We also consider that the beams are evenly spaced in the angular space of interest with inter-beam distance $2/N_R$ at the UE  ($2/N_T$ at the BS): ${\psi}^c_{l+1} = {\psi}^c_l+2/N_R$ (${\phi}^c_{s+1} = {\phi}^c_s+2/N_T$).
Denote also $\Psi_\ell = [{\psi}^c_\ell-1/N_R,{\psi}^c_\ell+1/N_R]$, $\ell=1,\ldots,L$ ($\Phi_s = [{\phi}^c_s-1/N_T,{\phi}^c_s+1/N_T]$, $s=1,\ldots,S$) as the intended coverage intervals of the UE (BS) beams.  The beamforming gain of beamformer $\mathbf{u}_\ell$ ($\mathbf{w}_s$)  at a direction within its intended coverage interval, i.e., $\psi \in \Psi_\ell$ ($\phi \in \Phi_s$) is typically much stronger than the gain of other beams:
	\begin{align}\label{Eq:desire_pattern}
	GR_\ell(\psi)\doteq|\mathbf{u}_\ell^{\dag}\mathbf{a}_R(\psi)|^2\gg |\mathbf{u}_j^{\dag}\mathbf{a}_R(\psi)|^2, ~\forall~\psi \in \Psi_\ell, \ell\neq j, \nonumber \\
	GT_s(\phi)\doteq|\mathbf{w}_s^{\dag}\mathbf{a}_T(\phi)|^2\gg |\mathbf{w}_p^{\dag}\mathbf{a}_T(\phi)|^2, ~\forall~\phi \in \Phi_s, s\neq p,
	\end{align}
	where $\mathbf{a}_R(\psi) \in {\mathbb C}^{N_R \times 1}$ and $\mathbf{a}_T(\phi) \in {\mathbb C}^{N_T \times 1}$ are the array response vectors of BS and UE, respectively.\footnote{Note that when the direction $\psi$ is close to the boundary of $\mathbf{u}_\ell$, e.g., $\phi_\ell$, the beamforming gain of $\mathbf{u}_{l+1}$ may be comparable to $\mathbf{u}_\ell$, which depends on the number of antennas and the beam synthesis technique used. Note also that although we have assumed DFT beams that have the peak gain equal to the number of antennas, the proposed beam search method to be presented can also be applied to other beams that satisfy the property in~\eqref{Eq:desire_pattern}. There is a rich literature~\cite{Xiao2016,zhang2017codebook,xiao2017codebook,xiao2018enhanced,fan2018flat} in designing beam synthesis techniques that produce beam patterns with flexible beam width and the properties specified by \eqref{Eq:desire_pattern}.}
	
\begin{figure*}[t]
	\centering
	\includegraphics[width=0.9\textwidth]{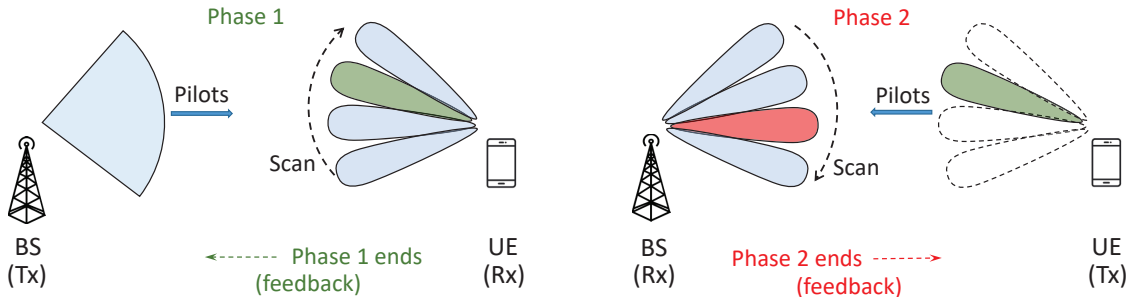}
	\caption{An illustration of two-phase search for beam alignment: In Phase 1, BS is the transmitter (Tx) while UE is the receiver (Rx). In Phase 2, BS is the Rx while UE is the Tx.  \label{fig_llinearSearch}}
\end{figure*}

With the DFT beams at both BS and UE, the total number of beam pairs, $S\times L$, is proportional to $N_T\times N_R$, which can be very large. Instead of performing a direct search over the $S\times L$ BS-UE beam pairs~\cite{hur2013millimeter}, we propose an alternative beam alignment protocol which we call two-phase spatial scanning, as illustrated in Fig.~\ref{fig_llinearSearch}. In the first phase, the BS transmits pilot signals in the downlink using a wide beam covering its entire angular interval of interest and the UE measures the pilots using beams in ${\cal L}$. Upon finding the best beam at UE, the UE sends a signal informing the BS that phase 1 has ended. In the second phase, the UE transmits pilot signals in the uplink using the beam identified in Phase 1 and the BS measures the pilots using beams from ${\cal S}$. Once the best BS beam is identified, the BS sends a signal to the UE that Phase~2 has ended. This completes
beam alignment. The two-phase protocol reduces the number of beams examined from $S\times L$ to $S+L$. We note that the two-phase protocol requires a dedicated feedback channel to exchange the completion message at each phase, which may be realised using low-frequency carriers. For instance, in 5G non-standalone (5G NSA) implementation, 5G radios use LTE channels to transmit control signalling. We also note that this two-phase protocol has been considered by several existing works in the literature, see~\cite{palacios2017tracking} for an example.

We assume a block fading model where the channel remains unchanged during spatial scanning and consider that the BS and UE are synchronised (see~\cite{liu2017design,barati2014dreictional} for possible mmWave synchronisation techniques). We also consider that spatial scanning is performed via pilots transmitted in a coherent frequency band, which can be a few tens of MHz or higher~\cite{rappaport2015wideband}. Denote $\mathbf{H}\in {\mathbb C}^{N_R \times N_T}$ as the channel matrix between BS and UE. Let the unit-norm vector $\mathbf{w}\in {\mathbb C}^{N_T \times 1}$ be the wide beam that the BS adopts in Phase 1 and define $\mathbf{u}^*$ and $\mathbf{w}^*$ as
\begin{equation}\label{Eq:genie_aid}
\mathbf{u}^* \doteq \arg\max_{\mathbf{u}_\ell\in {\cal L}}|\mathbf{u}^{\dag}_\ell \mathbf{H}\mathbf{w}|,~\mathbf{w}^* \doteq \arg\max_{\mathbf{w}_s\in {\cal S}}|\mathbf{u}^{*\dag} \mathbf{H}\mathbf{w}_s|.
\end{equation}
Clearly, $\mathbf{w}^*$ and $\mathbf{u}^*$ are the best BS-UE beams that will be selected in the two-phase search when there is no noise. In practice beam search has to find a pair of BS and UE beams which achieve good spectrum efficiency compared to that of the best beams $\mathbf{w}^*$ and $\mathbf{u}^*$, using noisy measurements provided by pilot signals. As we will discuss in the next section, this outcome can be obtained by adopting a Bayesian criterion for the unknown channel.

 \section{Beam Search with Iterative Deactivation and Beam Shifting}\label{Sec:algorithm}
 The beam search problems in Phase 1 and Phase 2 are equivalent, i.e., identifying the best beam from a set of noisy measurements of the candidate beams. The Iterative Deactivation and Beam Shifting (IDBS) algorithm is designed to solve this general beam search problem, and will be applied to both Phase 1 and Phase 2. In this section, we will focus on Phase 1 to explain the details of IDBS. We note that in both phases it is the {\it receiver} that does the beam search. The transmitter (Tx) sends pilot symbols and does not need to know the beam selections made at the receiver (Rx). In phase~1 the Rx is the UE, in phase~2 the Rx is the BS.

\subsection{IDBS Sketch and Signal Model}
IDBS is an iterative algorithm and it works by maintaining two sets of beams, namely the \emph{active set} and the \emph{inactive set}. In each iteration, training measurements are collected at the Rx using beams in the active set.\footnote{Note that IDBS is applied to both Phase 1 and Phase 2. Thus in each phase, there are multiple IDBS iterations.} After collecting the measurements (in each iteration), both the active and inactive set are updated. The updates are performed at the Rx by moving the active beams that are deemed to be relatively weak in comparison with the other active beams from the active set to the inactive set - we call this step ``beam deactivation".  Thus the active set reduces progressively in  size until a stopping point is reached and the algorithm terminates. At this point final beam selection is made.  Note that sometimes the active set can grow, since the algorithm allows for a deactivated beam to be placed back in the active set under exceptional  circumstances- we call this step ``inactive beam restoration". IDBS thus has the following key components which will be discussed in detail below: 1) Beam deactivation; 2) Inactive beam restoration; 3) Stopping criteria and 4) Final beam selection.

Of these four components, the one at the heart of the algorithm is ``Beam deactivation". Here a Bayesian prior, reflecting the
uncertainty of the unknown channel gain, is used in calculating the a posteriori probability that one beam has a stronger gain than another.
Beams which are seen as unlikely to be the strongest according to a given probability threshold are deactivated and no longer measured. The idea is that at high SNR the best beam returns large measurement values with high probability and weak beams are thus quickly deactivated, saving search time. On the other hand, at low SNR, the same criterion will act to retain all the beams, as their measured values are all similar, until sufficient measurements are taken to discriminate between them.

In most cases the algorithm will  stop as a result of only one beam being left in the active set. However, it will also halt when there are two beams left covering adjacent angular intervals,  to deal with the case that the dominant path of the channel falls between two beams.

The benefit of the beam search outlined in this section comes from finding the best Rx beam more quickly than with exhaustive search. This not only saves time, but also energy since the Tx can stop sending pilots earlier. This benefit would disappear in a fully digital system where all the receiver beams could be tested simultaneously. However, in analog beamforming (as considered at mm-wave) the beams are switched, and hence tested one at a time.

Before describing the four key components of the algorithm in detail, we first define some necessary notations for Phase 1 as follows.

Let ${\cal L}(t) \subseteq {\cal L}$ be the active set at iteration $t~(t=1,2,\cdots)$, and $n_\ell(t)$ be the accumulated number of pilot symbols for beam $\ell$ from iteration $1$ to iteration $t$. At each iteration, each active beam is sampled once such that at iteration $t$, the number of accumulated pilot symbols for all the active beams are the same: $n_\ell(t) = tn_0$, $\ell \in {\cal L}(t)$, where $n_0$ is the number of pilot symbols in each measurement.
Denote $h_\ell\doteq \mathbf{u}^{\dag}_\ell \mathbf{H}\mathbf{w}$ as the effective channel after Tx beam $\mathbf{w}$ and Rx beam $\mathbf{u}_\ell$  and let $\mathbf{s}_\ell(t) \in {\mathbb C}^{n_\ell(t) \times 1}$ be the vector containing all the pilot symbols used to measure beam $\mathbf{u}_\ell$ from iteration~1 to iteration~$t$ in Phase 1, with $\|\mathbf{s}_\ell(t)\|_2^2 = P_Tn_\ell(t)$ and $P_T$ being the transmission power. The measured signal at Rx, $\mathbf{y}_\ell(t) \in {\mathbb C}^{n_\ell(t) \times 1}$, can then be represented as:
\begin{equation}\label{Eq:received_signal_t}
\mathbf{y}_\ell(t) = h_\ell\mathbf{s}_\ell(t) + \mathbf{z}_\ell(t).
\end{equation}
Here  $\mathbf{z}_\ell(t)$ is the additive noise vector whose elements are assumed to be independent circularly symmetric Gaussian random variables with known variance $\sigma^2$.
The combined matched filter output is
\begin{equation}\label{Eq:matched_filter_t}
r_\ell(t) = \frac{1}{\|\mathbf{s}_\ell(t)\|^2_2} \mathbf{s}^{\dag}_\ell (t)\mathbf{y}_\ell(t).
\end{equation}
It can be seen that  $r_\ell(t)$ follows a complex Gaussian distribution with mean $h_\ell$ and variance $\sigma^2/\|\mathbf{s}_\ell(t)\|^2_2$, therefore
\begin{equation}\label{Eq:chi_square_t}
T_\ell(t) \doteq {2n_\ell(t)P_T|r_\ell(t) |^2}/{\sigma^2}\sim \chi_2^2(\eta_\ell(t)),
\end{equation}
i.e., $T_\ell(t) $ follows a non-central chi-square distribution with Degree of Freedom (DoF) equal to 2 and a non-centrality parameter $\eta_\ell(t)$ given by:
\begin{equation}\label{Eq:noncentrality}
\eta_\ell(t) = {2n_\ell(t) P_T|h_\ell|^2}/{\sigma^2}.
\end{equation}

\subsection{Beam Deactivation}\label{Sec:probability}
At each iteration, the set of beams to be deactivated, defined as ${\cal L}_E(t) \subseteq {\cal L}(t)$, are identified according to a single test function $f(\cdot)$, which makes a comparison between pairs of beams and deactivates the one that appears to be much weaker. As mentioned briefly in Section I, this test function is obtained according to the posterior probability that a beam has a smaller underlying channel gain than another. 

We now present the derivation of the test function $f(\cdot)$ by deriving the posterior probability $\tilde{f}(\cdot)$ that beam $\ell$ has a stronger underlying effective channel than beam $j$, given observations $T_\ell(t)$ and $T_j(t)$.  $\tilde{f}(T_\ell(t),T_j(t))$ can be represented as:
\begin{align}
\tilde{f}(T_\ell(t),T_j(t))& \doteq  \Pr\left \{ |h_{\ell}|>|h_j| \big | (T_{\ell}(t),T_j(t))\right \}\nonumber\\
&= \Pr\left \{ \frac{\eta_{\ell}}{n_\ell(t)}>\frac{\eta_j}{n_j(t)} \big | (T_{\ell}(t),T_j(t))\right \}  \label{Eq:post_prob_1}\\
&=\Pr\left \{\eta_\ell>{\eta_j}\big | (T_\ell(t),T_j(t))\right \} , \label{Eq:post_prob_2}
\end{align}
where \eqref{Eq:post_prob_1} follows from \eqref{Eq:noncentrality}, and \eqref{Eq:post_prob_2} is due to the fact that $n_j(t) = n_\ell(t)$, $~\forall~\ell, j \in {\cal L}(t)$.

To calculate~\eqref{Eq:post_prob_2}, we need to know $p( \eta_\ell,\eta_j|T_{\ell}(t),T_j(t) )$, the probability density of $(\eta_\ell,\eta_j)$ given observations $(T_{\ell}(t),T_j(t))$. Since there is no prior knowledge as to which beam is stronger than another one, nor the strengths of the effective channels for the beams to be searched in initial alignment, it is reasonable to treat them equally.  For this reason, we suppose that the non-centrality parameters are i.i.d. random variables that are uniformly distributed in $[0,\eta^+]$. We further suppose that $\eta^+\rightarrow \infty$ as $\eta_\ell$ can be large due to the possibility of high-SNR and large $n_\ell$. We note that the uniform prior with $\eta^+\rightarrow \infty$ is often referred as ``improper prior", which is commonly adopted in statistical inference when there is lack of knowledge as to which prior distribution to choose.
	
With the above assumptions, $p( \eta_\ell,\eta_j|T_{\ell}(t),T_j(t) )$ can be represented as
\begin{align}
p( \eta_\ell,\eta_j|T_{\ell}(t),T_j(t) ) & =  p( \eta_\ell|T_{\ell}(t))p(\eta_j|T_j(t) )  \label{Eq:density_final_1} \\
& = \frac{p(T_\ell(t)|\eta_\ell)p_{\eta_\ell}}{p(T_{\ell}(t))}  \frac{p(T_j(t)|\eta_j)p_{\eta_j}}{p(T_j(t))},  \label{Eq:density_final}
\end{align}
where $p(\eta_j|T_j(t) )$ is the probability density of $\eta_j$ given observation $T_j(t) $, $p_{\eta_j} = 1/\eta^+$ is the prior probability density of $\eta_j$, $p(T_j(t)|\eta_j)$ is the likelihood and $p(T_j(t))$ is the density of $T_j(t)$. Eq.~\eqref{Eq:density_final_1} follows from the fact that the measurements of the various beams are independent as they are collected at different times, while \eqref{Eq:density_final} is due to Bayes' theorem.

Now denote $g_{\eta}(x)$ as the probability density function of $\chi_2^2(\eta)$,
\begin{equation}\label{eq:chi2_density}
g_{\eta}(x) = \frac{1}{2}\exp(-\frac{x+\eta}{2})I_0(\sqrt{\eta x}),
\end{equation}
where $I_0(x)$ is the modified Bessel function of the first kind of zero-order~\cite{32133}.

Since the non-centrality parameters $\eta_\ell$'s are drawn uniformly from $[0,\eta^+]$, it follows that
\begin{align}
p( \eta_\ell | T_\ell ) &= \frac{p(T_\ell|\eta_\ell)p_{\eta_\ell}}{p(T_{\ell})}  = \frac{1/\eta^+ g_{\eta_\ell}(T_\ell)}{ p(T_\ell)} \nonumber\\
&= \frac{1/\eta^+ g_{T_\ell}(\eta_\ell)}{1/\eta^+ \int_0^{\eta^+} g_{T_\ell}(\eta) d\eta} \label{Eq:rr}\\
& =   \frac{ g_{T_\ell}(\eta_\ell)}{1 - Q_1(\sqrt{T_\ell}, \sqrt{\eta^+})}
\end{align}
where $Q_1$ is the Marcum Q-function~\cite{32133}, and to obtain \eqref{Eq:rr} we have used the reciprocity relationship $g_{\eta}(x) = g_x(\eta)$.

To reflect the dependence on $\eta^+$, we rewrite $\tilde{f}$ of \eqref{Eq:post_prob_2} as $\tilde{f}_{\eta^+}$ which can be obtained as:
\begin{align}
&\tilde{f}_{\eta^+}(T_\ell(t),T_j(t)) = \int\limits_{\eta^+\geq\eta>\nu\geq0} \frac{g_{T_\ell(t)}(\eta)} {\lb 1 - Q_1(\sqrt{T_\ell(t)} , \sqrt{\eta^+})\rb}\times \frac{g_{T_j(t)}(\nu) } {\lb 1 - Q_1(\sqrt{T_j(t)}, \sqrt{\eta^+ })\rb} d\eta d\nu.
\end{align}

By taking $\eta^+ \rightarrow \infty$, we can obtain the test function $f = \lim_{\eta^+ \rightarrow \infty} \tilde{f}_{\eta^+}$ as follows:
 \begin{align}
f(T_\ell(t),T_j(t)) &= \int_{\eta>\nu} g_{T_\ell(t)}(\eta) g_{T_j(t)}(\nu)   d\eta d\nu \label{Eq:finite_prob}.
\end{align}
which is the posterior probability of error under an improper (uniform) prior.  

The test function $f(T_k(t),T_j(t))$ is compared against a fixed threshold $\alpha \in (0,1)$. By comparing all possible pairs in the active set ${\cal L}(t)$, the deactivation set ${\cal L}_E(t)$ can be identified as:
\begin{equation}\label{Eq:eliminate_set}
	{\cal L}_E(t) = \{j: j\in {\cal L}(t), \max_{k\neq j, k\in{\cal L}(t)}f(T_k(t),T_j(t)) >\alpha\}.
\end{equation}

We now present some properties of $f(T_\ell(t),T_j(t))$  which are used to simplify the deactivation step.
\begin{theorem}\label{theorem_monoto}
Function $f(T_\ell(t),T_j(t))$ is monotonically increasing with respect to $T_\ell(t)$ and monotonically decreasing with respect to $T_j(t)$.
\end{theorem}

\begin{theorem}\label{theorem_2}
The set of beams being eliminated at iteration $t$, i.e., ${\cal L}_E(t)$ in \eqref{Eq:eliminate_set}, can be equivalently identified by:
	\begin{equation}\label{Eq:eliminate_set_simplified}
		{\cal L}_E(t) = \{j: j\in {\cal L}(t)\setminus\ell^*(t), f(T_{\ell^*(t)}(t),T_j(t)) >\alpha\},
	\end{equation}
where $\ell^*(t) = \arg \max_{\ell\in {\cal L}(t)} T_\ell(t)$. Further,
${\cal L}_E(t) = \emptyset$
if $f(T_{\ell^*(t)}(t),T_{\ell^-(t)}(t))\leq \alpha$, where $\ell^-(t) = \arg \min_{\ell\in {\cal L}(t)} T_\ell(t).$
\end{theorem}
Theorem~\ref{theorem_monoto} is proved in Appendix~\ref{Proof_theorem1}. Theorem~\ref{theorem_2} follows directly from the monotonicity of function $f(\cdot)$ in Theorem~\ref{theorem_monoto}. Theorem~\ref{theorem_2} shows that it is sufficient to compare beam $\ell^*(t) = \arg\max_{\ell\in{\cal L}(t)}T_{\ell}(t)$, the one with the strongest measurement, against the other $|{\cal L}(t)|-1$ beams.

\begin{lemma} \label{Lemma_infinite_sum}
Function $f(x,y)$ has the following closed-form representation:
\begin{align}
&f(x,y) = 1-\frac{1}{2}\exp \left (-\frac{x}{2} - \frac{y}{4}\right ) \sum_{m=1}^{+\infty} \left(\frac{1}{2}\right)^m L_m(-\frac{y}{4})\sum_{k=0}^{m-1}\frac{(\frac{x}{2})^k}{k!}. \label{eq:close_form}
\end{align}
where $L_m(a)$ is the Laguerre polynomial of order $m$.
\end{lemma}
\begin{proof}
See Appendix~\ref{Appendix_proof_infinite_sum}.
\end{proof}
Lemma~\ref{Lemma_infinite_sum} provides a tractable way to numerically evaluate function $f(x,y)$, in which a finite number of summands can be used to approximate the infinite sum of \eqref{eq:close_form}. The number of summands $M$ required to reach good accuracy depends on the value $x$ and $y$. $M$ can be high when both $x$ and $y$ are large and comparable, making it difficult to compute in real time.

Fortunately, due to the monotonicity properties proved in Theorem~\ref{theorem_monoto}, evaluating $f(x,y)$ and making a comparison with the threshold $\alpha$ to identify ${\cal L}_E(t)$ can be alternatively realised by comparing $y$ to a critical value $\tau_\alpha(x)$, where $f(x,\tau_\alpha(x)) = \alpha$. Because $f(x,y)$ monotonically decreases with respect to $y$, $f(x,y)>\alpha$ is equivalent to $y<\tau_\alpha(x)$. With $\tau_\alpha(x)$, it can be seen that set ${\cal L}_E(t)$ given by~\eqref{Eq:eliminate_set_simplified} can be equivalently identified as follows:
\begin{equation}\label{Eq:eliminate_set_simplified_2}
{\cal L}_E(t) = \{j: j\in {\cal L}(t)\setminus \ell^*(t),  T_j(t) < \tau_\alpha(T_{\ell^*(t)})\}.
\end{equation}

Fig.~\ref{fig:critical_value} plots the critical value $\tau_\alpha(x)$ as a function of $x$ for four different $\alpha$ values. It can be seen that $\tau_\alpha(x)=0$ when $x$ is smaller than some value $x_\alpha$. When $x>x_\alpha$, $\tau_\alpha(x)$ is smooth and can be approximated by a small-degree polynomial and thus does not require to be calculated in real-time. As shown in Fig.~\ref{fig:critical_value}, for $\alpha = 0.9$, the error between a quadratic fit (when $x>x_\alpha$) obtained using least squares (magenta curve) and $\tau_\alpha(x)$ (black curve) is very small. Similar small degree polynomial approximations can be made for wider intervals of $x$ and for other choices of $\alpha$. As an alternative for practical implementations, $\tau_\alpha(x)$ can also be stored as a look-up table for a range of $x$ values.

When a proper threshold $\alpha$ is used or the pre-beamforming SNR is not very low, the probability of the strongest beam being deactivated is small. To see this, we have provided analysis in Appendix~\ref{sec_accept}. The results there show that 1) the probability of the strongest beam ever being deactivated is small when the pre-beamforming SNR is -10 dB across a range of relatively high values of $\alpha$, and the probability of deactivating the strongest beam will be even lower at higher SNRs due to monotonicity; 2) when the pre-beamforming SNR is -20 dB, the probability of the strongest beam being deactivated is still small provided that $\alpha$ is sufficiently high ($\alpha = 0.99$), see Table~\ref{table:UnionBnd} in Appendix~\ref{sec_accept}. Note that the probability of deactivating the strongest beam will be even lower at higher SNRs when $\alpha=0.99$.

\begin{figure}
		\centering
		\includegraphics[width=0.6\textwidth]{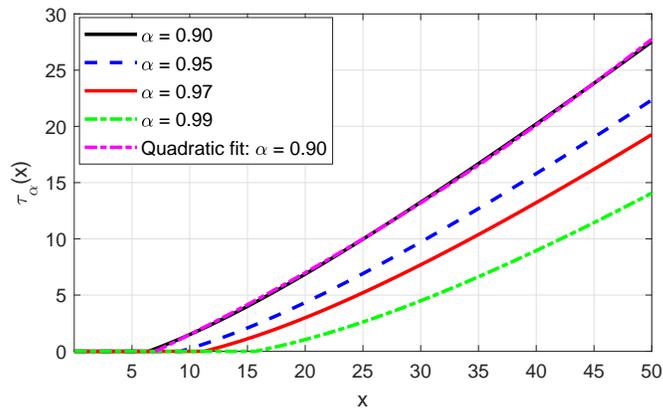}
	\caption{Rejection region area below $\tau_\alpha(x)$ vs $x$. }\label{fig:critical_value}
\end{figure}

\subsection{Inactive Beam Restoration}\label{Sec:restoration}
From Table~\ref{table:UnionBnd} in Appendix~\ref{sec_accept}, it can also be seen that when $\alpha$ is relatively low (eg., $\alpha = 0.90$) and the pre-beamforming SNR is very low (e.g., -20 dB), there is some chance that the best beam is deactivated.  If this occurs, it tends to happen when $t$ is small, because the upper bound of the probability that the best beam for the first time being deactivated at iteration $t$ decreases exponentially with $t$, see Fig.~\ref{fig_MaxUnionRej} in Appendix~\ref{sec_accept}. If the strongest beam is deactivated, the active set may contain a set of beams that all have similar but weak post-beamforming channel gains. The consequence would be $f(T_\ell(t),T_j(t))$ in later iterations would not meet the threshold $\alpha$, because the underlying channels are all comparable (note that $f(x,y)\big|_{x=y}=0.5$), leading to excessive training time wasted, as well as a failure to find the best beam.

This issue is addressed by the  Inactive Beam Restoration step which restores the best beam to the active set if it happens to be deactivated. This step is thus performed when an inactive beam has the strongest matched filter output, i.e., $\ell^*\neq \ell^*(t)$, where $\ell^* = \arg \max_{l\in {\cal L}}|r_\ell(t)|$ is the beam that has the strongest matched filter output from all candidate beams ${\cal L}$, including the active beams and the inactive beams. (We remind the reader that $\ell^*(t)$ is the beam that has the strongest matched filter output from the active beams ${\cal L}(t)$, which is not necessarily the same to $\ell^*$.) If this condition is triggered, beam $\ell^*$ is required to be measured again such that its accumulated number of pilot symbols is $n_{\ell^*}(t) = n_0t$ before it can be moved back to the active set. We pause to note that this is due to the requirement that all active beams must have the same number of measurements such that beam deactivations can be realised by low-complexity operations using a single look-up table.\footnote{It is possible to drop the restriction that $n_\ell(t) = n_0t$, $\forall \ell \in {\cal L}(t)$ and to derive a generalisation to \eqref{eq:close_form}. However, under this circumstance, it is required to store many look-up tables for different pairs of $(n_\ell(t),n_j(t))$, or to compute the generalised form of \eqref{eq:close_form} in real time. This would either require much larger memory space or higher computational complexity and would render the algorithm less attractive from a practical point of view.}  After collecting the additional measurements, beam $\ell^*$ is then added to the active set as: ${\cal L}(t) \leftarrow {\cal L}(t) \cup \{ \ell^*\}$.

As we will show in Section~\ref{Sec:restoration_results}, the restoration step can reduce average overhead and improve spectrum efficiency when relatively low values of $\alpha$ (such as $\alpha = 0.90$)
are used. The inclusion of this step increases the chance that the best beam is retained in the active set, thus increasing the chance that the best beam is selected at the end of the search and also that the weak beams are deactivated more rapidly. The performance improvement from beam restoration becomes negligible for larger values of $\alpha$ as the chance that the strongest beam is deactivated is low. Inactive beam restoration makes IDBS more robust to the choice of $\alpha$, which may be chosen differently to achieve different tradeoffs between overhead and spectrum efficiency (see Sec.~\ref{Sec:Numerical}).

\subsection{Stopping Criteria}
At each iteration, IDBS checks against the following three stopping conditions to decide whether it should terminate.
\begin{enumerate}[I.]
\item Only one beam remains active:
	\begin{equation}\label{Eq:stop_2}
		|{\cal L}(t)|=1.
	\end{equation}
\item Two adjacent beams are active
	\begin{align}\label{Eq:stop_5}
				|{\cal L}(t)|=2,~|\ell^*(t) - j| = 1,~\forall~j\in {\cal L}(t), j\neq \ell^*(t),
	\end{align}
	and they have comparable statistics $T_\ell(t)$'s:
	\begin{equation}\label{Eq:stop_4}
	{\max_{\ell\in {\cal L}(t)}T_\ell(t)}/{\min_{\ell\in {\cal L}(t)}T_\ell(t)}<2
	\end{equation}
\item There is not enough training time to collect measurements for an additional iteration:
	\begin{equation}\label{Eq:stop_1}
		N(t)\doteq \sum_{l\in {\cal L}}n_\ell(t) > N_1 - {n_0}|{\cal L}(t)|,
	\end{equation}
	where $N_1$ is the maximum number of pilot symbols, available to Phase 1.
\end{enumerate}
Condition I in~\eqref{Eq:stop_2} means that all but one beam has been deactivated from further consideration. Condition II in \eqref{Eq:stop_5}--\eqref{Eq:stop_4} is an early stopping criterion which is proposed to deal with the possibility that a dominant path falls close to the boundary of two adjacent beams.\footnote{Note that Condition II requires both \eqref{Eq:stop_5} and \eqref{Eq:stop_4} to be met. If only \eqref{Eq:stop_5} is met, IDBS will not terminate.} Condition III in~\eqref{Eq:stop_1} means that the remaining time is not enough to scan the active beams for an additional iteration. Without the early stopping condition, i.e., Condition II, IDBS can take longer search time to reach Condition I or even use up all the available time, as we will show in Section~\ref{Sec:Toy_example}.

\subsection{Decision Rule}
The fact that IDBS terminates due to Condition II also provides an opportunity for IDBS to select a beam that is better than any of the two active beams, namely the beam centred at the boundary of the two adjacent beams. Motivated by this consideration, we propose the following decision rule for IDBS.
When IDBS terminates, a beam is selected for subsequent data transmission. The selected beam could either be one from the original codebook used for spatial scanning, or a shifted version of one of the original beams, according to the following decision rules:
\begin{enumerate}[A.]
\item If iterations are terminated when Condition I or Condition III is met, then choose beam $\ell^*(t)$ as the beam for subsequent use.
\item If iterations are terminated when Condition II is met, then shift beam $\ell^*(t)$ by half a beam width towards the other active beam $j$. For ULA, the shifted beamforming vector can be obtained by~\cite{Xiao2016}:
\begin{equation}\label{Eq:beam_switch}
\mathbf{u}'_{\ell^*(t)} = \mathbf{u}_{\ell^*(t)} \odot [1,e^{-i\frac{2\pi d}{\lambda}\delta_{\psi}},\ldots,e^{-i\frac{2\pi d}{\lambda}(N_R-1)\delta_{\psi}}]^T,
\end{equation}
where $\delta_{\psi} = \psi_{j} - \psi^c_{\ell^*(t)}$ and $\odot$ is the Hadamard product.
\end{enumerate}

We emphasise that the beam shifting operation due to Condition II requires no extra training time but provides an opportunity for IDBS to select a better beam than any of the original beams used for spatial scanning, i.e., an opportunity to even outperform the best beams from the original codebook given by~\eqref{Eq:genie_aid}. This benefit will be confirmed by numerical results in Sec.~\ref{Sec:Toy_example} and Sec.~\ref{Sec:Numerical}.

The IDBS algorithm is summarised in Table~\ref{table:algorithm}.  Note that IDBS in Table~\ref{table:algorithm} is applied first to Phase 1 and then to Phase 2. The same threshold $\alpha$ is adopted for the two phases. Suppose the maximum training time available for beam alignment, measured by the number of pilot symbols, is $N^{+}$. To apply IDBS to both phases, the maximum overhead of Phase 1 is set to $N_{1}<  N^+ - n_0S$ so that there is enough time for the Tx to scan its entire codebook, recalling that $S$ is the number of Tx candidate beams. The maximum overhead available to Phase~2 varies as it depends on how much resource is spent at Phase~1. However, it is at least $N^+-N_{1}$.

We finally note that IDBS requires one feedback each at the end of Phase 1 and Phase 2. As discussed earlier, these messages can be exchanged using low-frequency carriers, similarly to 5G NSA where low-frequency channels are used to transmit control signalling. The termination messages are not required to contain the beam index identified in each phase nor the indices of the deactivated beams. This is because the beam measurements and deactivation are performed by the Rx (e.g., UE in Phase 1) while the Tx is only required to keep transmitting the pilots (e.g., BS in Phase 1), as explained at the start of Section III. Therefore, the number of bits to feedback in each message can be as small as 1 bit. IDBS requires fast beam switching, due to the iterative beam examination where each beam is scanned using a short pilot sequence. However, this fast beam switching can be supported by state-of-the-art technologies. For instance, IBM has reported beam switching speeds of $< 4$ ns~\cite{IBM_beamswitch}.

\begin{table}[t]
	\begin{center} \caption{IDBS for Beam Alignment} \label{table:algorithm}
		\resizebox{0.8\textwidth}{!}{\begin{tabular}{l}
				\hline
				{\bf Input}: Beam codebook ${\cal L}$; Total overhead budget $N_1$; Threshold $\alpha$; Look-up table $(\tau_\alpha(y),y)$ \\
				{\bf Initialisation}: $t\leftarrow 0$; ${\cal L}(t) \leftarrow {\cal L}$, $N(t)\leftarrow 0$, $n_\ell(t) \leftarrow 0$, $\forall~\ell\in {\cal L}(t)$, flag $\leftarrow$ 0 \\
				{\bf While}  flag = 0\\
				
				\quad {\bf Step~1)}: collect one additional measurement for each of the beams \\
				\quad \quad~~~$~l\in {\cal L}(t)$, compute $r_\ell(t)$ and $T_\ell(t)$ as in \eqref{Eq:matched_filter_t} and \eqref{Eq:chi_square_t}\\
				\quad {\bf Step~2)}: \\
				\quad \quad 2.1) Identify ${\cal L}_E(t)$,the set of beams to deactivate, according to \eqref{Eq:eliminate_set_simplified_2}\\
				\quad \quad 2.2) Update the active set ${\cal L}(t) \leftarrow {\cal L}(t) \setminus {\cal L}_E(t)$\\
				\quad {\bf Step~3)}: $\ell^* = \arg \max_{\ell\in {\cal L}}|r_\ell(t)|$ and $\ell^*(t) =\arg \max_{\ell \in {\cal L}(t)} |r_\ell(t)|$.  \\
				\quad \quad \quad\quad\quad If $l^*\neq l^*(t)$: \\
				\quad \quad \quad\quad\quad\quad \quad 3.1) Collect extra measurements for beam $\ell^*$ such that $n_{\ell^*}(t) = n_0 t$ \\
				\quad \quad \quad\quad\quad\quad\quad 3.2) Update the active set as ${\cal L}(t) \leftarrow {\cal L}(t) \cup \{ \ell^*\}$\\
				\quad {\bf Step~4)}: Update $N(t)$ and Check Stopping Criteria\\
				\quad \quad \quad\quad ~~ I. If $N(t)+|{\cal L}(t)|> N_1$ \& $|{\cal L}_A(t)|>1$, flag $\leftarrow$ 1\\
				\quad \quad \quad\quad ~~ II. If $|{\cal L}_A(t)|=1$, flag $\leftarrow$ 2\\
				\quad \quad \quad\quad ~~ III. If \eqref{Eq:stop_5} and \eqref{Eq:stop_4} are satisfied, flag $\leftarrow$ 3\\
				{\bf Decision}: \\
				\quad\quad\quad\quad If flag = 1 or flag = 2, choose beam $\mathbf{u}_{\ell^*(t)}$\\
				\quad\quad\quad\quad If flag = 3, shift beam $\mathbf{u}_{\ell^*(t)}$ according to \eqref{Eq:beam_switch}
				\\\hline
		\end{tabular}}
	\end{center}
\end{table}

\section{Examples to Explain Condition II of the Stopping Criteria and Beam Shifting of the Decision Rule }\label{Sec:Toy_example}
In this section, we use two simple examples to explain the motivation of adopting Condition II of the stopping criteria and the beam shifting operation as presented in Section~\ref{Sec:algorithm}. In the two examples, we consider a single-path channel, a UE equipped with a single antenna and a BS with a ULA of $N_T = 32$ antennas.  In this setup, beam scanning only occurs at BS in Phase 2. The BS codebook consists of DFT vectors as given by~\eqref{eq:Tx_DFT}.
We also consider that the angular interval of BS is a $60^\circ$-sector, i.e., $\tilde{\phi} \in [-30^\circ,+30^\circ]$ or equivalently $\phi \in [-1/2,1/2]$, thus the codebook has $S=16$ beams. With this setup, $\phi_s^c= -\frac{1}{2} + \frac{(s-1/2)}{16}$, $s = 1,\ldots,16$, and the effective channel gain
\begin{equation}\label{Eq:single_path}
|h_s|^2 = |\mathbf{H}\mathbf{w}_s|^2 = |\gamma|^2|\mathbf{w}^{\dag}_s\mathbf{a}_T(\phi)|^2 \propto GT_s(\phi),
\end{equation}
where we recall that $GT_s(\phi)\doteq |\mathbf{w}^\dag_s\mathbf{a}_T(\phi)|^2 $ is the beamforming gain of $\mathbf{w}_s$ at angle $\phi$. Fig.~\ref{fig:two_examples} (a) and (b) plots $GT_s(\phi)$ with respect to beam index $s$, in Example 1 and Example 2, respectively. It can be seen that in Example 1, beam 9 has a much stronger channel than all other beams, as the channel path falls nearly in the centre of its intended coverage interval. In comparison, in Example 2, both beam 9 and 10 have much stronger channels than the other beams, as the channel path falls close to the boundary between the intended coverage intervals of beam 9 and 10. However, the beamforming gains of beam 9 and 10 in Example 2 are much weaker than that of beam 9 in Example 1.

The results to be presented in this section are averaged over $5\times 10^3$ trials with $\mathbf{H}$ fixed for Example 1 and Example 2 but noises generated randomly. The training budget is set to $N^+ = 1024$ symbols and $\alpha = 0.97$.

Fig.~\ref{fig:toy_overhead} is the average of the total number of pilot symbols when IDBS terminnates, with and without Condition II as specified in~\eqref{Eq:stop_5} and \eqref{Eq:stop_4}. The horizontal axis is the pre-beamforming SNR, which equals $\frac{P_T|\gamma|^2}{\sigma^2}$.
As expected, in Example 1, there is no discernible difference with and without Condition II of the stopping criterion. In Example~2, significantly lower overhead is observed across SNRs when Condition II is included in the stopping criteria.

Fig.~\ref{fig:one_dominant_beam} illustrates how training time is spent on each beam in Example 2. The vertical axis is the average number of pilot symbols spent on each beam, i.e., ${\mathbb E}\{n_s(t)\}$, at termination. It can be seen that for both the pre-beamforming SNR values presented, $-20$ dB for Fig.~\ref{fig:one_dominant_beam}(a) and $0$ dB for Fig.~\ref{fig:one_dominant_beam} (b), IDBS without Condition II spends much longer time attempting to make a decision between beam 9 and beam 10, which have comparable effective channel strengths.

\begin{figure}
	\centering
	\begin{minipage}{.48\linewidth}
		\centering
		\includegraphics[width=1\textwidth]{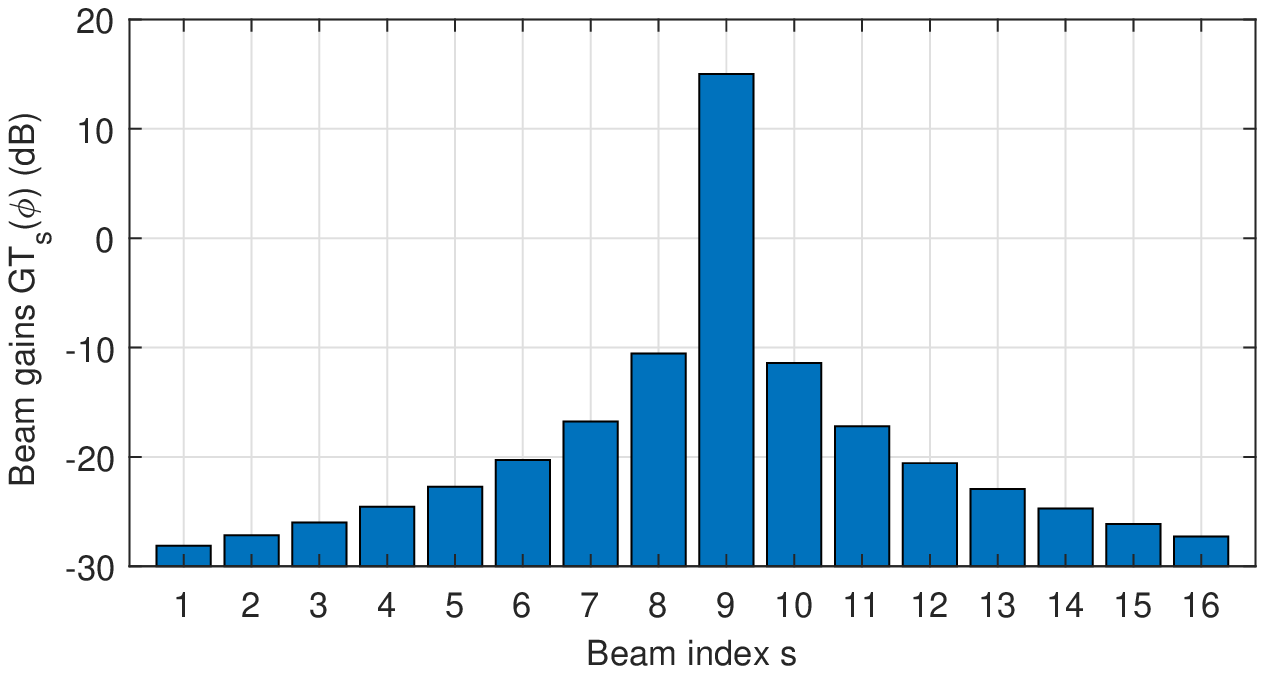}	\\
		(a)	
	\end{minipage}
	\begin{minipage}{.48\linewidth}
		\centering
		\includegraphics[width=1\textwidth]{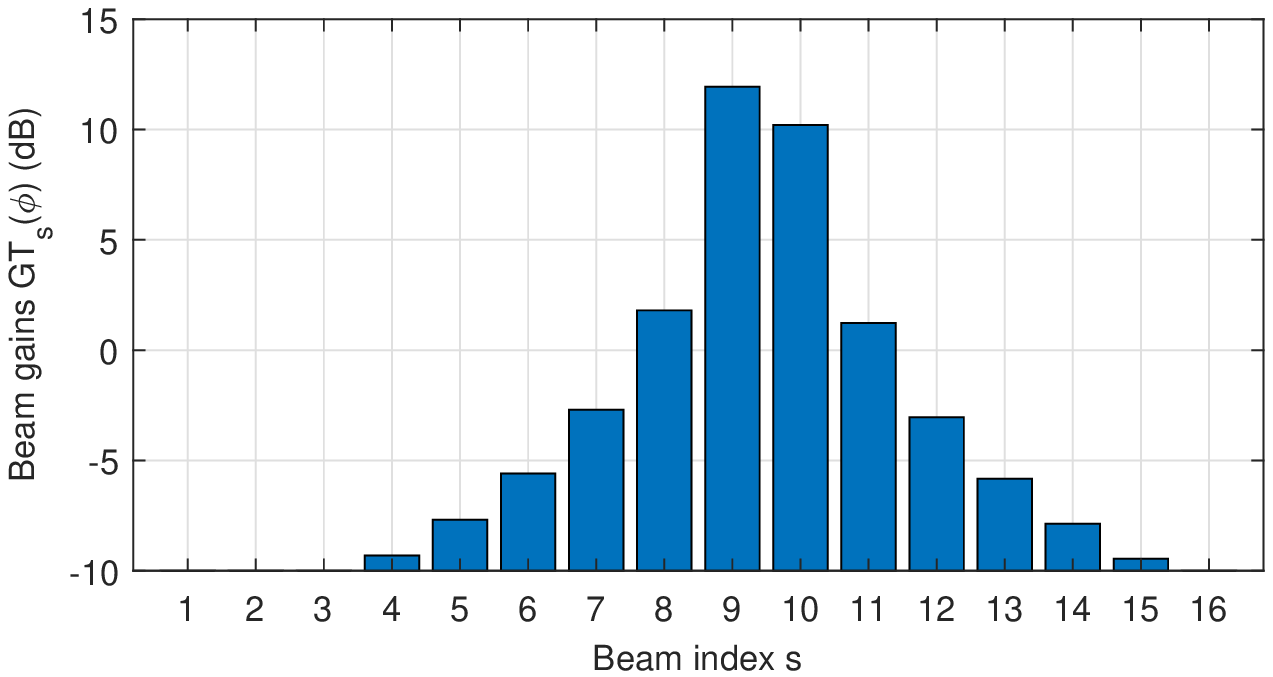}\\
		(b)
	\end{minipage}
	\caption{(a) Example 1: Only one dominant beam (beam 9). $\sin(\phi) = 0.0141$ (b) Example 2:   Two comparable beams (beam 9 and 10). $\sin(\phi) =0.0297$. The intended coverage of beam 9 is $[0,0.0313]$. The intended coverage of beam 10 is $[0.0313,0.0626]$.}\label{fig:two_examples}
\end{figure}

\begin{figure}
		\centering
		\includegraphics[width=0.55\textwidth]{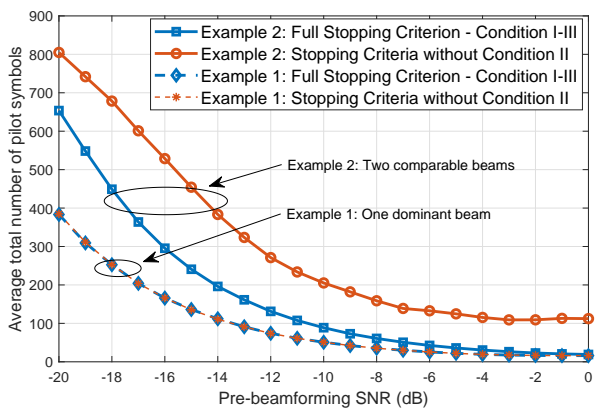}
	\caption{Average total number of pilot symbols when different stopping criteria are adopted. }\label{fig:toy_overhead}
\end{figure}

\begin{figure}[t]
	\centering
	\begin{minipage}{.48\linewidth}
		\centering
		\includegraphics[width=1\textwidth]{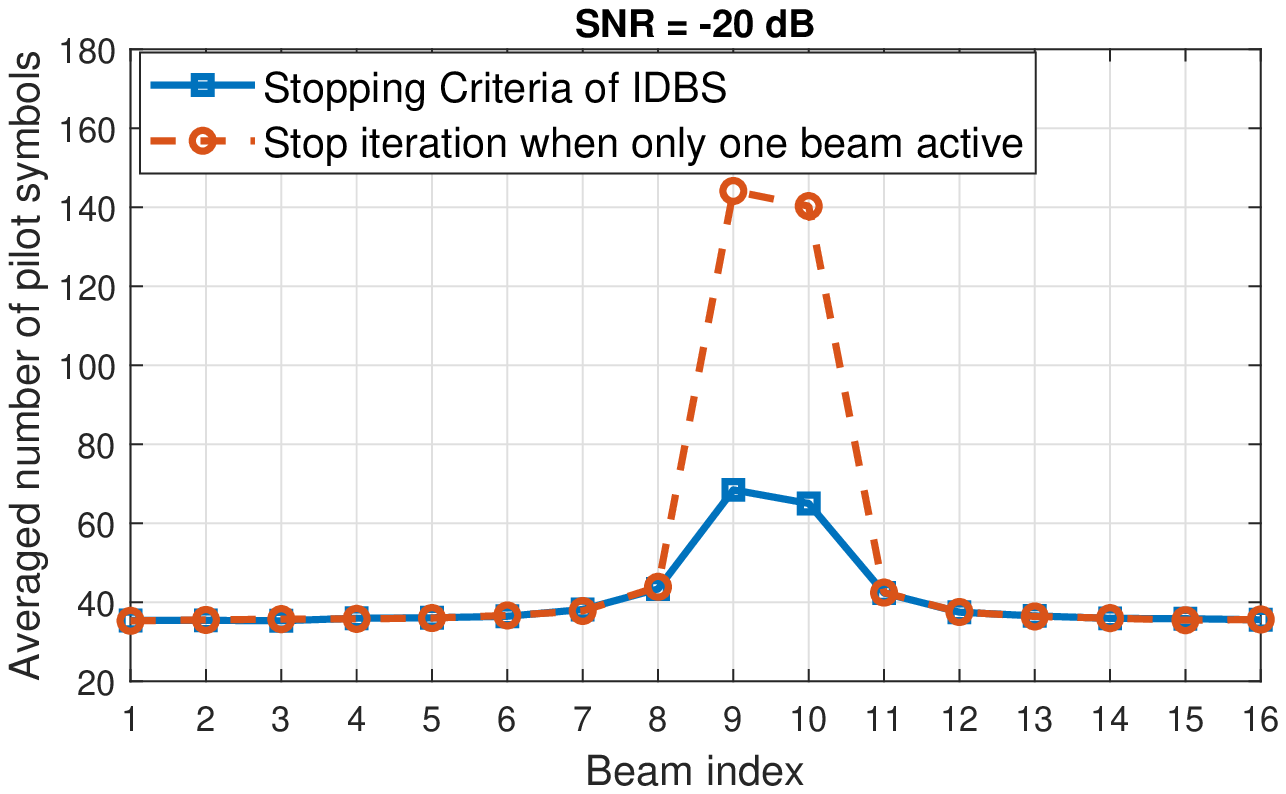}	\\
		(a)	
	\end{minipage}
	\begin{minipage}{.48\linewidth}
		\centering
		\includegraphics[width=1\textwidth]{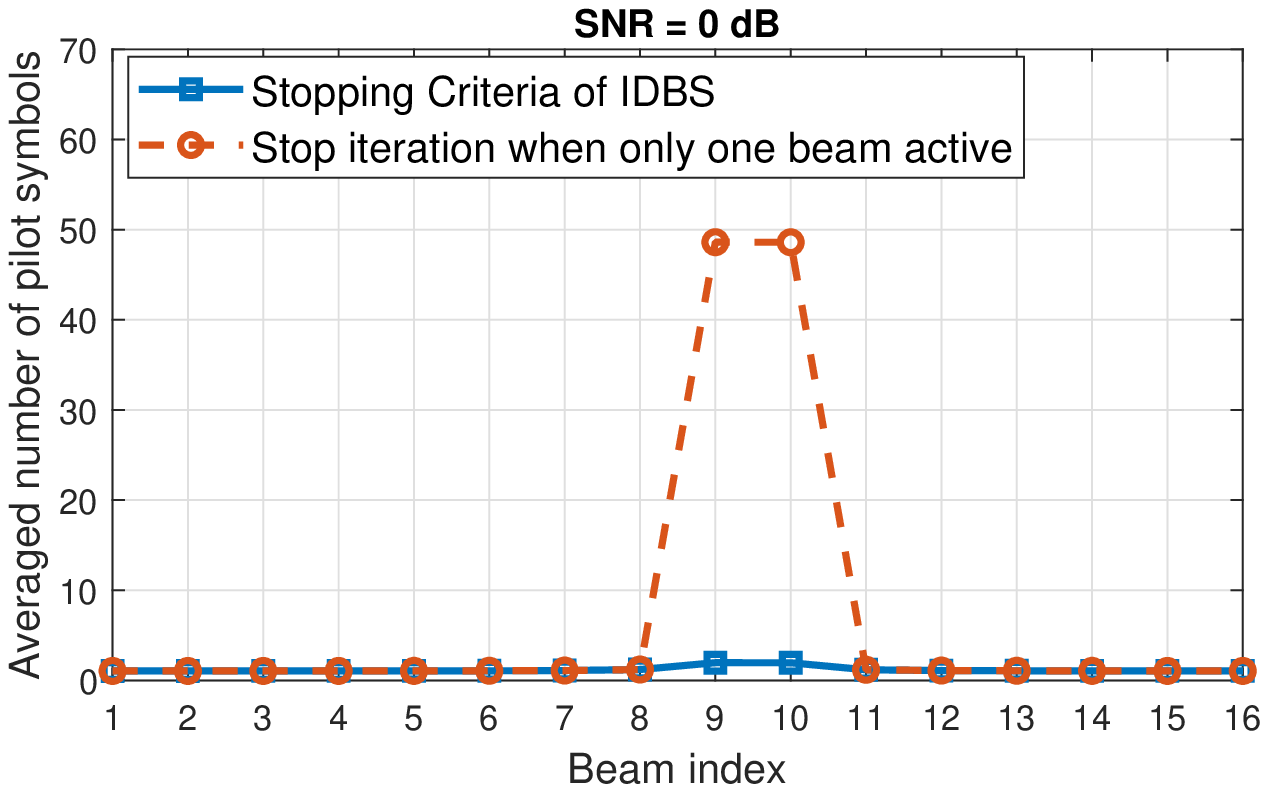}\\
		(b)
	\end{minipage}
	\caption{Example 2 - two comparable beams. Average number of pilot symbols when different stopping criteria are adopted: (a) SNR = -20 dB; (b) SNR = 0 dB}\label{fig:one_dominant_beam}
\end{figure}

Fig.~\ref{fig:toy_overhead} and Fig.~\ref{fig:one_dominant_beam} show that Condition II indeed leads to significant reduction of training time when the dominant path falls near the boundary of two adjacent beams.  We next show that the reduction of training time, due to the adoption of Condition II, does not come at the price of sacrificing the performance of subsequent data transmission. In fact, the combined use of Condition II and the beam shifting rule can lead to noticeable performance improvement in scenarios like Example 2.

Fig.~\ref{Fig:rate_Example2} presents the average spectrum efficiency after beam search in Example 2. The curve labeled as ``IDBS"
is obtained when the full version IDBS as described in Section~\ref{Sec:algorithm} is adopted for beam selection. The curve labeled as ``IDBS: no beam shifting" assumes that IDBS is adopted with all the three stopping conditions (Condition I-III) but without the beam shifting given. In other words, it always chooses the active beam that has the strongest measurement for data transmission. The comparison between IDBS with and without beam shifting is to demonstrate the effectiveness of beam shifting. The curve labeled as ``IDBS: no Condition II no beam switching" is obtained assuming IDBS with only Condition I and Condition III as the stopping criteria and without beam shifting. The final curve labeled as ``Best beam from scanning codebook" corresponds to the rate calculated using~\eqref{eq:rate} below assuming the best beams given in~\eqref{Eq:genie_aid} are selected. For all these four cases, the spectrum efficiency is calculated assuming that the same transmission power is adopted by beam training and data transmission. Therefore,  the spectrum efficiency is given by:
\begin{equation}
R = \log_2\left(1+\frac{P_T\big |\hat{\mathbf{u}}^\dag \mathbf{H}\hat{\mathbf{w}}\big |^2}{\sigma^2}\right)
\label{eq:rate}
\end{equation}
where the UE beam $\hat{\mathbf{u}}$ equals one (since a single antenna is assumed at UE) and $\hat{\mathbf{w}}$ is the beam selected under each scheme considered.

\begin{figure}
	\centering
	\includegraphics[width=0.55\textwidth]{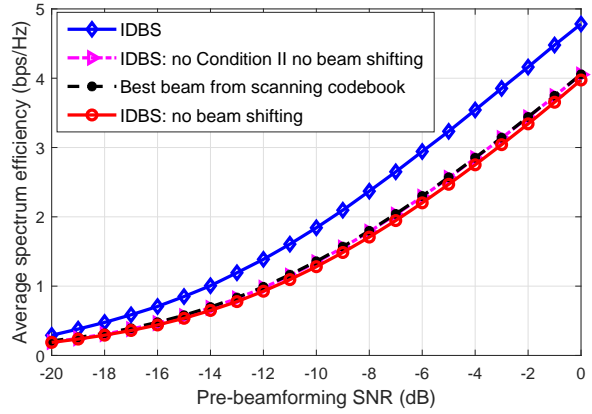}
	\caption{Example 2: Average spectrum efficiency after beam alignment, with or without beam switching.}
	\label{Fig:rate_Example2}
\end{figure}

It can be seen from Fig.~\ref{Fig:rate_Example2} that ``IDBS with no Condition II and no beam shifting" does have slightly higher average spectrum efficiencies than ``IDBS without beam shifting". This is because the former spends much more time on measuring beam 9 and beam 10, and thus can choose  the slightly better beam more reliably.  However, the performance differences are not significant across the SNR range and they come from significantly longer training time.

In can also be seen from Fig.~\ref{Fig:rate_Example2} that both ``IDBS with no Condition II and no beam shifting" and ``IDBS without beam shifting" have similar average spectrum efficiencies to that of the best beams from the original codebook, across the entire SNR range.  This means that IDBS can almost always find the best beam in this example across the entire SNR range. IDBS improves upon ``IDBS without beam shifting" and even outperforms the best beams from the original codebook. This is because the beam shifting operation in IDBS allows it to choose a beam that is outside the original codebook used for scanning and which offers a stronger beamforming gain.

To summarise, the examples in this section demonstrate that the introduction of Condition II can help reduce training overhead when the dominant path falls close to the boundary of two adjacent beams. The beam shifting operation performed after stopping on Condition II can provide better rate performance.\footnote{It is noted that in the rare occasion that two or more paths falling to the centres of adjacent beams, beam shifting may lead to worse rate performance. However, for IDBS to suffer from this problem, not only must this situation occur, but the paths must also have comparable path strengths such that neither of the beams that cover these two paths is deactivated. This makes it even rarer and thus is not considered in the design of IDBS.}

\section{Numerical Results}\label{Sec:Numerical}
In this section, we evaluate the performance of IDBS and investigate the impact of parameter choices on its performance.  Throughout this section, we suppose that ULAs are equipped at both BS and UE, and that the number of antennas at BS is $N_T = 64$ and the number of antennas at UE is $N_R = 16$. The angular interval of interest for BS is $[-30^{\circ},+30^{\circ}]$, i.e., $\phi \in [-1/2,1/2]$, and the angular interval of interest for UE is $[-90^{\circ}, 90^{\circ}]$, i.e., $\psi \in [-1,1]$. This setup corresponds to the initial beam alignment stage where a UE is synchronised to a BS, whose coverage is a $60^{\circ}$-sector. The wide beam $\mathbf{w}$ used by the BS to cover the $60^{\circ}$-sector in phase 1 is synthesised using the algorithm in~\cite{fan2018flat}. The narrow beams used by the BS and UE are DFT beams with inter-beam distances $2/N_T$, $2/N_R$ and peak gains $N_T$ and $N_R$, respectively. Because the BS covers interval $\phi \in [-1/2,1/2]$, there are $S=32$ narrow beams in the BS codebook. The UE codebook has $L=N_R=16$ beams because the UE angular interval of interest is $\psi \in [-1,1]$. In what follows, we assume the same transmit power in beam training and data transmission (when evaluating the spectrum efficiency). We also consider the same noise variance $\sigma^2$ at BS and UE and $n_0=1$.

Following~\cite{liu2017Jsac}, we consider two scenarios, namely line-of-sight (LOS) and non-line-of-sight (NLOS) scenarios. In the LOS scenarios, a Rician channel model is considered where there is one dominant path from angle $\tilde{\phi}$ and $\tilde{\psi}$ with respect to the BS and UE. The Rician ${\mathcal K}$-factor (i.e., the ratio of the power of the dominant path to the sum of the power of the scattering components) is set to $13.2$ dB~\cite{muhi2010modelling} and both  $\tilde{\phi} \in [-30^\circ,+30^\circ]$ and $\tilde{\psi} \in [-90^\circ,90^\circ]$ are uniformly distributed.
In the NLOS scenairo, the channel is modelled as the sum of $I$ paths, each with $\tilde{\phi}_i$ and $\tilde{\psi}_i$ uniformly drawn from $[-30^\circ,+30^\circ]$ and $[-90^\circ,90^\circ]$, respectively. Each path is again assumed to be Rician, with ${\mathcal K}$-factor set to $6$ dB~\cite{samimi201628}. The number of paths is $I = \max\{1,\zeta\}$, where $\zeta$ is a Poisson random variable with mean $1.8$ and the power fractions of the $I$ paths are generated by the method in~\cite{Akdeniz2014}. In what follows, the results presented are obtained by averaging over $2\times 10^4$ random channel realisations. The average pre-beamforming SNR is defined as the ratio between the average sum power from all paths and the average noise power at the receiver.

\subsection{Impact of the threshold $\alpha$}
We first investigate the impact of the threshold $\alpha$ on the training overhead of IDBS and the spectrum efficiency. Two maximum overhead levels $N^+ =\{3072,1024\}$ are considered. In either case, the maximum budget for Phase 1 is set to $N_1 = N^+ - 32\times 4$ to ensure that the beam searching time left for Phase 2 is not too short.

Fig.~\ref{fig::LOS_avg}(a) and (b) respectively present the average overhead  and the average spectrum efficiency in the LOS case. The results for the NLOS case are not presented for brevity, as the insights are the same as in the LOS case.

\begin{figure*}[t]
	\centering
	\includegraphics[width=0.99\textwidth]{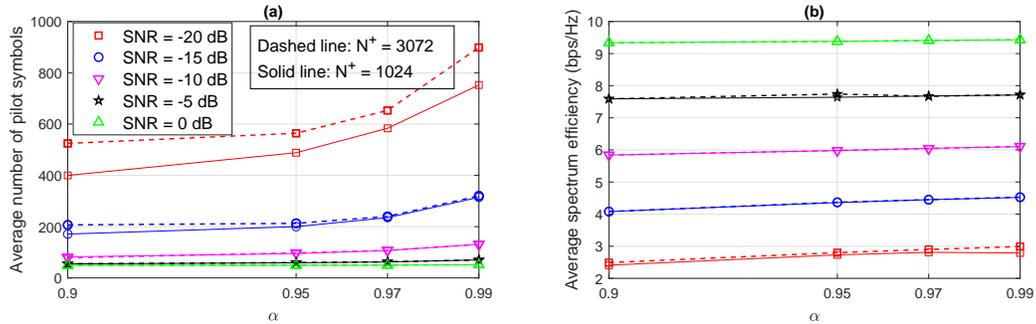}
	\caption{Impact of $\alpha$ on: (a) Average number of pilot symbols; (b) Average spectrum efficiency. Dashed lines: $N^+ = 3072$; Solid lines: $N^+ = 1024$.}
	\label{fig::LOS_avg}
\end{figure*}

As can be seen from Fig.~\ref{fig::LOS_avg}(a), the average overhead is lower when the average pre-beamforming SNR is higher, at all $\alpha$ values. As expected, a higher $\alpha$ increases the average overhead at both $N^+$ values. This overhead increase is more significant at lower SNRs. For instance, when $N^+ = 3072$, increasing $\alpha$ from $0.90$ to $0.99$ increases the average overhead by $(898.6/524.5 - 1) = 71\%$ when the SNR is -20 dB, while only by $(51.1/49.1-1) = 4\%$ when the SNR is 0 dB.

A higher $\alpha$ also tends to provide better spectrum efficiency, as one would expect and also shown in Fig.~\ref{fig::LOS_avg}(b). Similar to the observations on the overhead, the improvement in spectrum efficiency is significant only at low SNRs. For higher SNRs such as 0 dB, the performance does not depend strongly on $\alpha$ as beam alignment decisions are made within a few iterations.

Note that as shown in Fig.~\ref{fig::LOS_avg}(b), when SNR = -20 dB and $N^+=1024$, increasing from $\alpha = 0.97$ to $\alpha = 0.99$ slightly reduces the average spectrum efficiency. This is because with $\alpha = 0.99$, it becomes harder for IDBS to deactivate beams, thus it increases the chance that all the available training time is used up, i.e., IDBS terminates due to Condition III in~\eqref{Eq:stop_1}. Table~\ref{table:LOS_frac_overhead} presents the fraction of time that IDBS terminates due to Condition III in~\eqref{Eq:stop_1} in either Phase 1 or Phase 2, for SNR = -20 dB. It can be seen that with $\alpha = 0.99$, there is a $33.5\%$ of chance that it is running out of training time. For $\alpha = 0.97$, this fraction is much lower (14.5\%). For the higher maximum overhead $N^+=3072$, this fraction remains small even when $\alpha = 0.99$, thus the average spectrum efficiency keeps improving by increasing $\alpha = 0.97$ to $\alpha=0.99$.

\begin{table}[t]
	\caption{LOS - Fraction of time the maximum budget is used: values outside parentheses - $N^+ = 3072$; values inside parentheses - $N^+ =  1024$ \label{table:LOS_frac_overhead}}
	\centering
	\small{
		\begin{tabular}{|c|c|c|c|c|}
			\hline
		$\alpha$	& 0.90 & 0.95 & 0.97 & 0.99 \\
			\hline
		-20 dB & 4.8\% (11.1\%) & 1.9\% (9.5\%) & 1.5\% (14.5\%) & 1.4\% (33.5\%) \\
		\hline
		\end{tabular}
	}
\end{table}

\subsection{Impact of Inactive Beam Restoration}\label{Sec:restoration_results}
In this subsection, we compare the performance of IDBS to that of an incomplete version of IDBS, i.e., IDBS without the Inactive Beam Restoration step, to understand the impact of the restoration step. Fig.~\ref{Fig:LOS_Max} presents the overhead performance and the average spectrum efficiency of IDBS (with restoration) and IDBS without restoration for $\alpha = 0.90$ and $\alpha = 0.95$ in the LOS case. The maximum overhead is set to $N^+=1024$.

For $\alpha = 0.90$, it can be seen from Fig.~\ref{Fig:LOS_Max} (a) that Inactive Beam Restoration reduces the average overhead of IDBS. The overhead reduction is more significant at low SNRs, as the chance that the best beam being deactivated is higher in these scenarios (see the analysis in Appendix~\ref{sec_accept}). Fig.~\ref{Fig:LOS_Max}~(b) shows that with $\alpha = 0.90$, the adoption of Inactive Beam Restoration also increases the average spectrum efficiency. Both key aspects of performance are seen to improve.

As also shown by Fig.~\ref{Fig:LOS_Max} (a) and (b), with $\alpha = 0.95$, the performance improvement by using Inactive Beam Restoration becomes less significant. This is because the algorithm is less likely to deactivate the best beam when $\alpha$ is high, as explained in Appendix~\ref{sec_accept}, thus restoration is triggered much less frequently.

\begin{figure}
	\centering
	\begin{minipage}{.48\linewidth}
		\centering
		\includegraphics[width=1\textwidth]{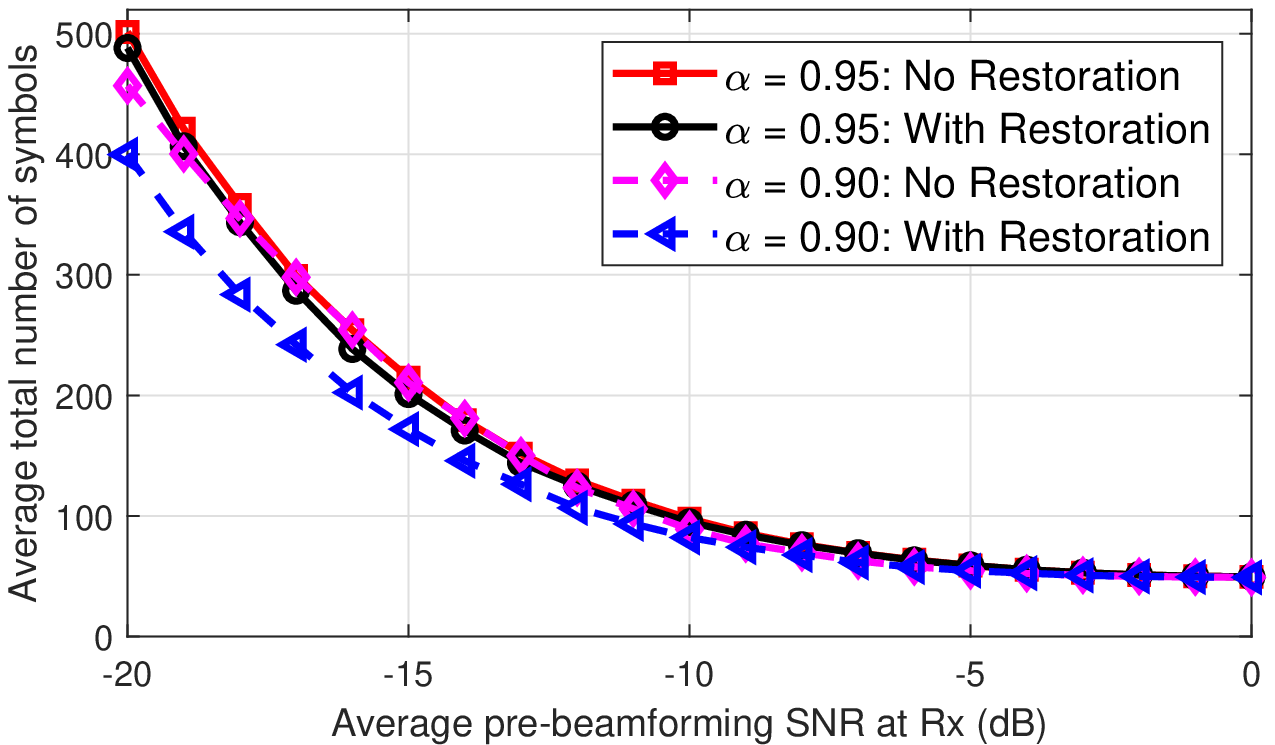}	\\
		(a)	
	\end{minipage}
	\begin{minipage}{.48\linewidth}
		\centering
		\includegraphics[width=1\textwidth]{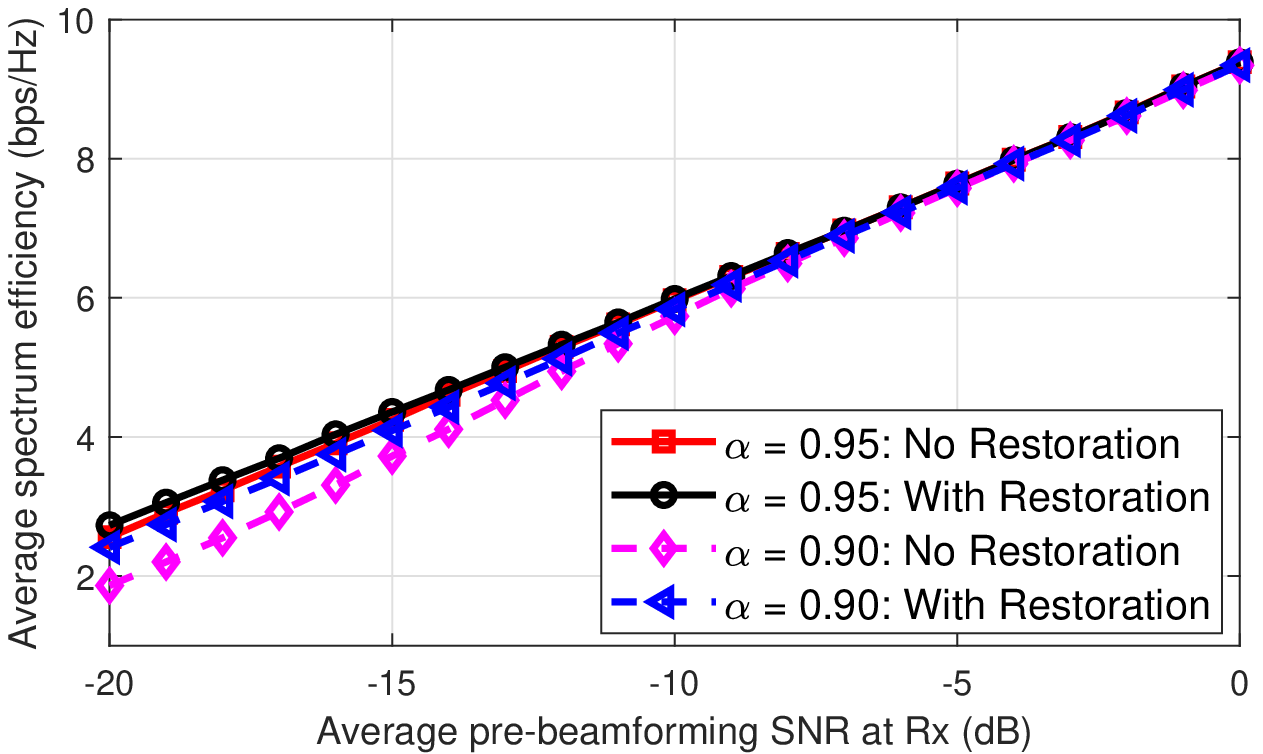}\\
		(b)
	\end{minipage}	
	\caption{Impact of the Inactive Beam Restoration on: (a) Average number of pilot symbols (overhead); (b) Average spectrum efficiency for $\alpha = 0.90$ and $\alpha = 0.95$, in LOS scenario.}\label{Fig:LOS_Max}
\end{figure}

\subsection{Performance Evaluation of IDBS}
In this subsection, we evaluate the average spectrum efficiency and the average overhead of IDBS in different SNRs and fading scenarios to demonstrate its capability of adapting to the unknown scenarios to achieve good tradeoffs between the two metrics. For illustration, we consider $N^+=1024$.

Fig.~\ref{Fig:meanRate_overhead_LOS}  and Fig.~\ref{Fig:meanRate_overhead_NLOS} plots the average spectrum efficiency versus the average overhead of IDBS in LOS and NLOS, respectively. Two pre-beamforming SNRs are considered: $-15$~dB and $-5$~dB. The four red points for IDBS in each plot correspond to $\alpha = [0.90,0.95,0.97,0.99]$ from left to right.  As benchmarks to the spectrum efficiency, we have added ``best beam pair of scanning codebook" (magenta dashed curve) and ``infinite-resolution beamforming" (black solid curve) in the figures. The ``best beam pair of scanning codebook" is obtained by searching over the original codebooks assuming the knowledge of channel, as given by~\eqref{Eq:genie_aid}. The ``infinite-resolution beamforming" benchmark is obtained by searching over the entire angular space at both the UE and the BS, i.e., $\psi \in [-1,1]$ and $\phi\in [-1/2,1/2]$:
\begin{align}\label{Eq:genie_aid_2}
&\mathbf{u}^* = \arg\max_{\psi \in [-1,1]}|\mathbf{u}^{\dag}(\psi) \mathbf{H}\mathbf{w}|,\nonumber\\
&\mathbf{w}^* = \arg\max_{\phi\in [-1/2,1/2]}|\mathbf{u}^{*\dag} \mathbf{H}\mathbf{w}(\phi)|,
\end{align}
where $\mathbf{u}(\psi)$ and $\mathbf{w}(\phi)$ are the UE and BS DFT beams centred at $\psi$ and $\phi$, respectively. Because the search of ``infinite-resolution beamforming" is in the continuous space as opposed to ``best beam pair of scanning codebook" that is over a set of discretised angles, the former has higher spectrum efficiencies than the latter, and is also an upper bound to the spectrum efficiency that can be achieved by directional analog beamforming.

\begin{figure}
	\begin{minipage}{.49\linewidth}
		\centering
		\includegraphics[width=1\textwidth]{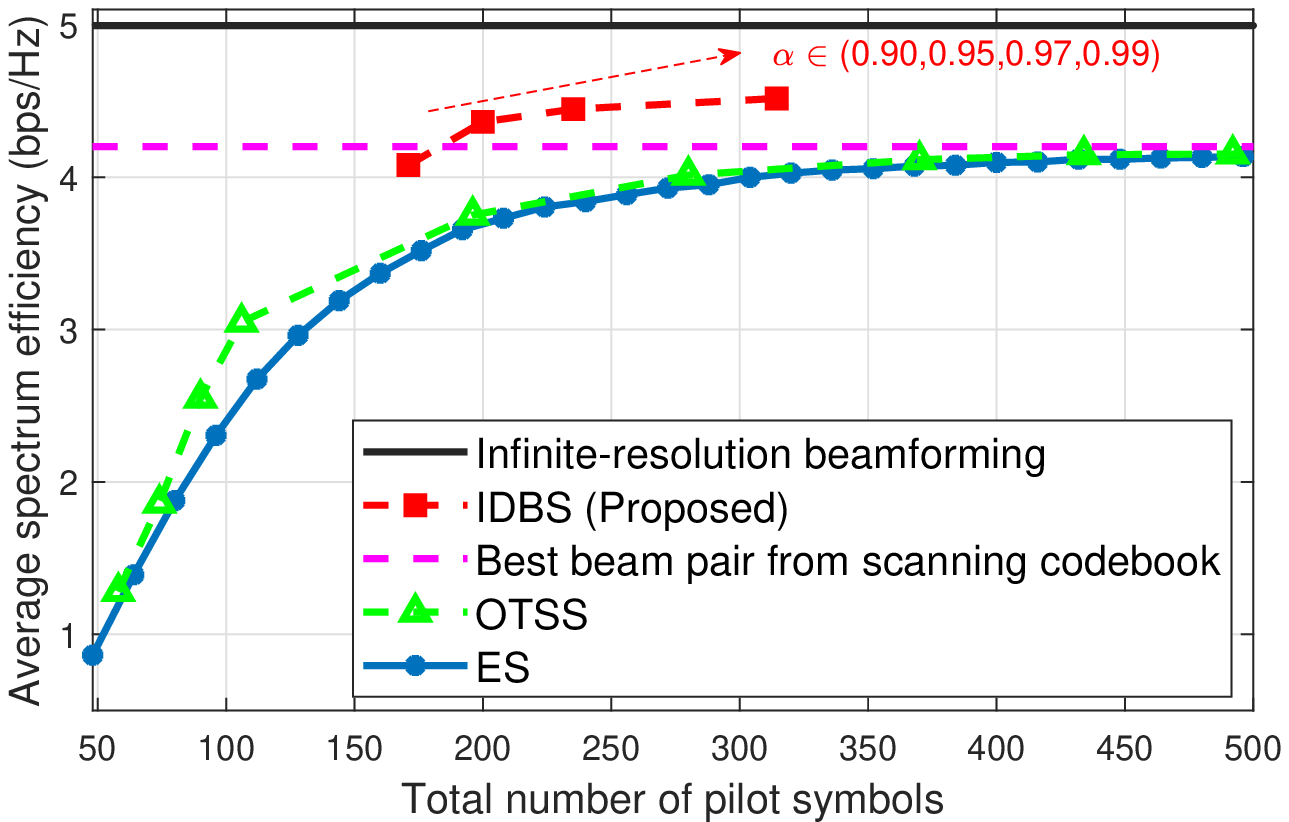}\\
		(a)
		\\
		\includegraphics[width=1\textwidth]{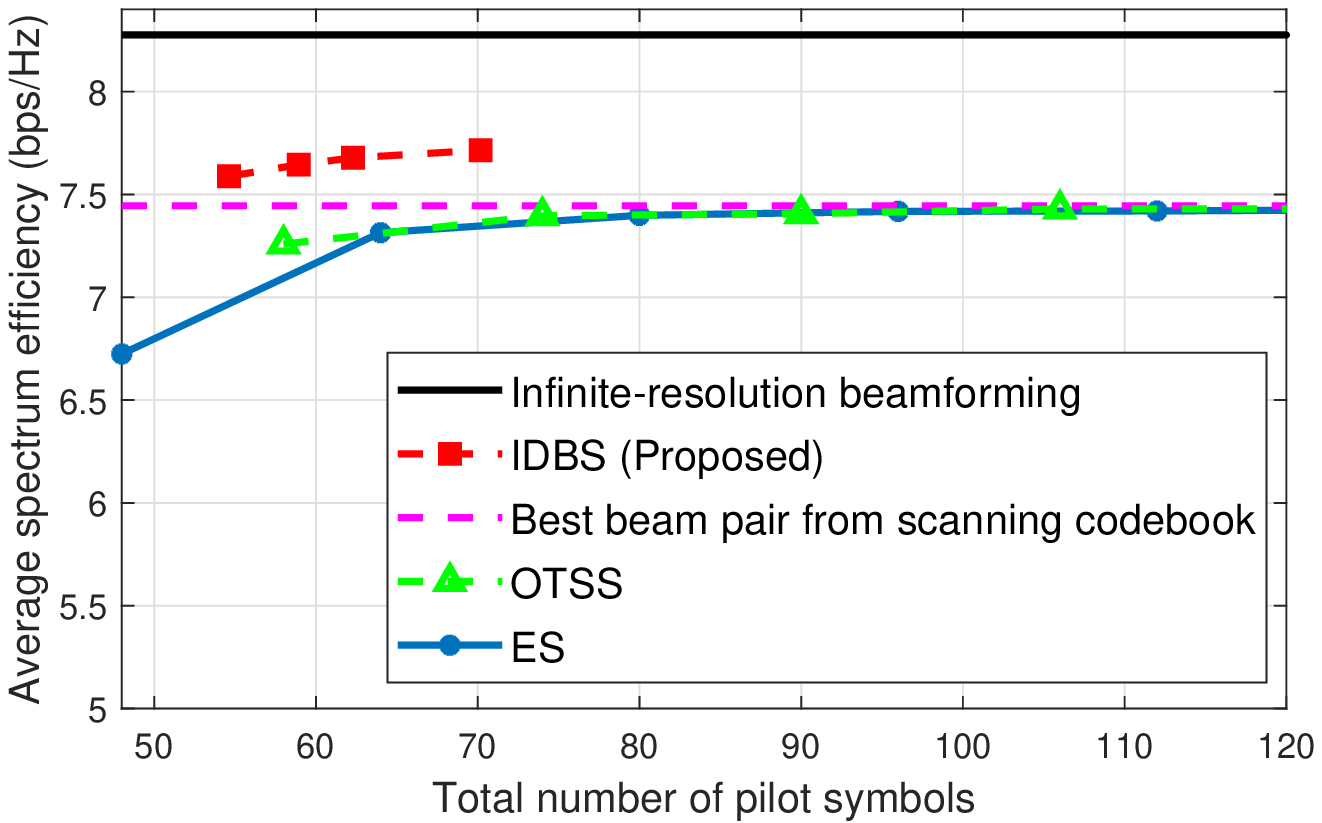}\\
		(b)
		\caption{LOS: Beam search performance tradeoff - average spectrum efficiency vs. overhead. (a) pre-beamforming SNR = -15 dB; (b) pre-beamforming SNR = -5 dB. The four red points of IDBS represent four values of $\alpha$, i.e., $\{0.90,0.95,0.97,0.99\}$ from left to right.  }
		\label{Fig:meanRate_overhead_LOS}
	\end{minipage}
	\hspace{1em}
	\begin{minipage}{.49\linewidth}
		\centering
		\includegraphics[width=1\textwidth]{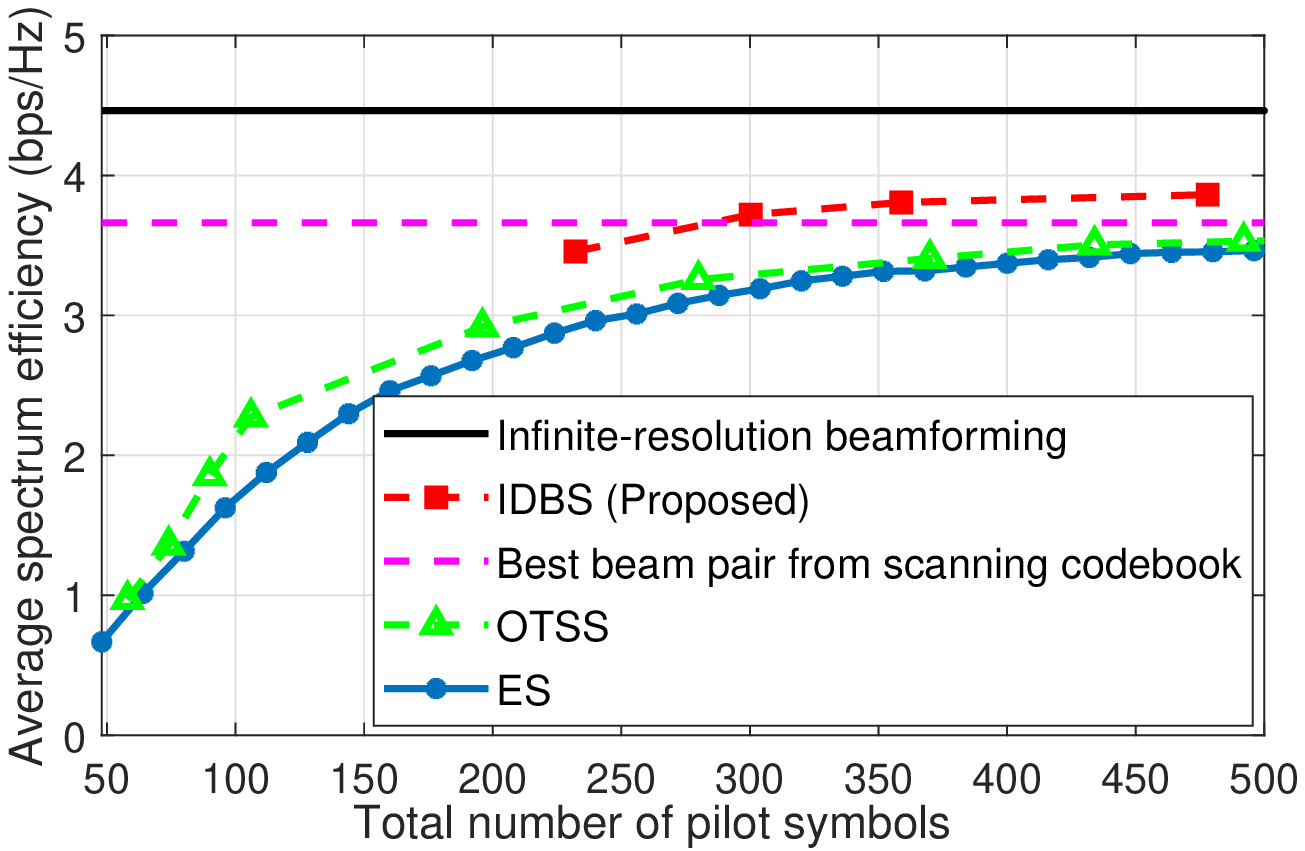}\\
		(a)\\
		\includegraphics[width=1\textwidth]{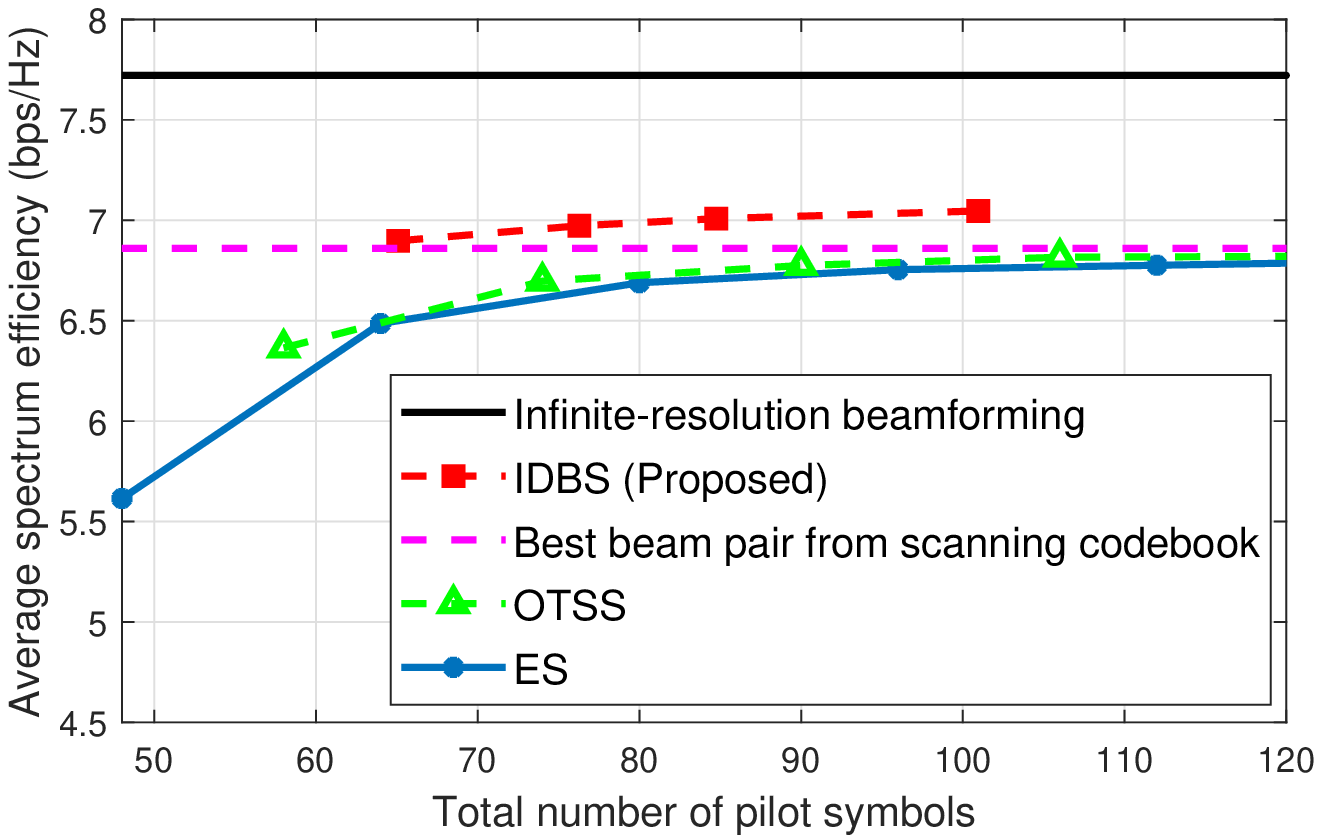}\\
		(b)\\
		\caption{NLOS: Beam search performance tradeoff - average spectrum efficiency vs. overhead. (a) pre-beamforming SNR = -15 dB; (b) pre-beamforming SNR = -5 dB. The four red points of IDBS represent four values of $\alpha$, i.e., $\{0.90,0.95,0.97,0.99\}$ from left to right.  }
		\label{Fig:meanRate_overhead_NLOS}
	\end{minipage}
\end{figure}

For comparison, we have also included the exhaustive search (ES) and the OTSS algorithm~\cite{min2019TWC} in Fig.~\ref{Fig:meanRate_overhead_LOS}  and Fig.~\ref{Fig:meanRate_overhead_NLOS}. Similar to IDBS, both the ES and OTSS are applied firstly to Phase 1 and then to Phase 2. (For ES, the beam search is performed in $S+L$ steps.) It is noted that for ES and OTSS, each point presented is obtained from the best time split between Phase 1 and Phase 2 for the corresponding overhead. This optimal time split varies from LOS to NLOS and from one SNR to another, and is obtained using a brute-force search over all feasible time splits assuming that the SNR and the LOS/NLOS fading statistics are known. We emphasise that the results for ES and OTSS are \emph{not achievable} in practice because the SNR and the LOS/NLOS labels are unknown in initial beam alignment (thus the optimal time split between phase 1 and phase 2 cannot be known either).

As can be seen from Fig.~\ref{Fig:meanRate_overhead_LOS}  and Fig.~\ref{Fig:meanRate_overhead_NLOS}, in the four (LOS/NLOS, SNR) scenarios presented, IDBS is able to use a suitable amount of training overhead to achieve good spectrum efficiency, by controlling a single parameter $\alpha$. For instance, in LOS, when $\alpha = 0.97$, IDBS uses an average of 235 and 62 symbols to achieve $4.4/5.0= 88\%$ and $7.7/8.2= 93.9\%$ to that of the optimal ``infinite-resolution beamforming" when SNR is $-15$ dB and $-5$ dB. Similar figures are seen also in NLOS scenarios.

Such good balances between overhead and spectrum efficiency are not achievable by ES and OTSS with a universally fixed overhead. To see this, let us consider first a fixed overhead at $300$ for ES. In LOS and SNR = -15 dB, see Fig.~\ref{Fig:meanRate_overhead_LOS}~(a), this overhead allows the ES to achieve $4.0/5.0 \approx 80\%$ of the ``infinite-resolution beamforming". However, the overhead at 300 is unnecessarily large for SNR = -5 dB, as there is hardly any improvement of spectrum efficiency by increasing the overhead beyond 100. For SNR = -5 dB and LOS, a more suitable overhead level appears around $60$, with which ES approaches ``best beam pair of scanning codebook" and achieves $7.3/8.2=89\%$ of the optimal ``infinite-resolution beamforming". But again, this overhead around 60 will be a rather poor choice for SNR = -15 dB, as the spectrum efficiency achieved will be very poor, i.e., only $1.4/5.0= 28\%$ of the ``infinite-resolution beamforming".  It can also be seen that the amount of overhead required by ES (and OTSS) differs from LOS to NLOS, with longer time required in NLOS. This fact further demonstrates the drawback of the non-adaptive approach with a universally fixed overhead.

We remark that IDBS can achieve higher spectrum efficiencies than the best beam pair from the scanning codebook. This is because IDBS, with beam shifting, is allowed to select a beam that is outside the original scanning codebook, providing it an opportunity to select a better beam. The benefit of beam shifting will be smaller if oversampled codebooks, with more beams placed closer to achieve higher spatial resolution, are used in spatial scanning. However, using oversampled codebooks will significantly increase the training overhead, as there are more beams to examine. A more effective way to gain higher spatial resolution is to refine the beams after the initial alignment, using beam refinement algorithms such as~\cite{zhu2018high}, instead of using an extremely fine resolution codebook for initial scanning.

We also remark that there are other alternative performance metrics to the average spectrum efficiency. For instance, an effective spectrum efficiency that takes account for the overhead spent in beam search: $R' = R\times (1-N_{\text{training}}/N_{\text{max}})$, may be used, where $N_{\text{max}}$ is the total number of symbols in a coherent block and $N_{\text{training}}$ is the number of symbols used for beam search. Another valid metric is the outage probability that captures the fraction of time that the achieved spectrum efficiency is lower than one of the benchmarks by a certain amount. Using these metrics do not change the overall insights drawn above, we thus omit them due to the page limit.

We finally note that the overhead of IDBS varies even for the same SNR due to random channel realisations and noise. The overhead variations tend to be more significant when the pre-beamforming SNR is lower. To see this, we present in Table~\ref{table:overhead_extreme} the 90-th percentile of the overhead, along with the average overhead, where the 90-th percentile is used to illustrate the overhead consumed by IDBS in ``worst-case" realisations. As can be seen from Table~\ref{table:overhead_extreme}, for SNR at $-5$ dB, LOS and $\alpha=0.95$, the 90-th percentile is 71 which is only $71-58.9=12.1$ larger than the average or $12.1/58.9=20.5\%$ larger in relative terms. The overhead consumed is relatively concentrated around the average. For SNR of $-15$ dB, LOS and $\alpha = 0.95$, the 90-th percentile of the overhead is $314$, which is $314-200=114$ larger than the average or $114/200 = 57\%$ larger in relative terms.

\begin{table}[t]
	\caption{IDBS overhead - values outside parentheses: 90-percentile of the overhead upon termination. Values inside parentheses:average overhead upon termination. Overhead is measured by the total number of pilots. $N^+=1024$. \label{table:overhead_extreme}}
	\centering
	\small{
		\begin{tabular}{|c||c|c|c|c|c|}
			\hline
			& 	& $\alpha=0.90$ & $\alpha=0.95$ & $\alpha=0.97$ & $\alpha=0.99$ \\
			\hline
			LOS & -15 dB & 263 (171.0) & 314 (200.0)  &  379 (235.1) & 507 (314.4) \\
			\hline
			LOS & -5 dB & 64 (54.6) &  71 (58.9) &  78 (62.3) &  91 (70.2) \\
			\hline
			\hline
			NLOS & -15 dB & 461 (232.2) & 606 (300.3)  &  778 (358.9) & 950 (478.0) \\
			\hline
			NLOS & -5 dB & 81 (65.1) &  105 (76.3) &  127 (84.8) &  166 (101.0) \\
			\hline
		\end{tabular}
	}
\end{table}

\section{Conclusions}\label{Sec:Conclusions}
In this paper, we presented a new  algorithm for mmWave beam search called Iterative Deactivation and Beam Shifting (IDBS). IDBS does not require any prior information such as SNR references and channel fading statistics in order to operate. It automatically adapts its overhead to the unknown SNR and fading statistics to obtain beam alignment performance close to that achieved by the best beams from the codebooks used for beam search, but with minimal overhead. Simulations over LOS and NLOS fading channel models extracted from NYU's measurements~\cite{Akdeniz2014} have confirmed that IDBS can achieve good balances between training overhead and beam search accuracy in different SNRs and in different fading scenarios. This makes the algorithm attractive for outdoor mmWave cellular mobile communications where the SNR and channel fading statistics vary significantly.

\color{black}
\appendices

\section{Proof of Theorem~\ref{theorem_monoto}}\label{Proof_theorem1}
\begin{align}
f(T_\ell(t),T_j(t)) &  = \int_{\eta_\ell>\eta_j}g_{T_\ell(t)}(\eta_\ell)g_{T_j(t)}(\eta_j)d\eta_\ell d\eta_j = \int^{+\infty}_{\eta}g_{T_\ell(t)}(\eta_\ell) \big [1-Q_{1}(\sqrt{T_j(t)},\sqrt{\eta})\big ]d\eta_\ell,
\end{align}
where $Q_{1}(a,b)$ is the Marcum Q-function with DoF = 1. Since $Q_{1}(a,b)$ is monotonically increasing with respect to $a$~\cite{5429099}, it follows that $f(T_\ell(t),T_j(t))$ is monotonically decreasing with respect to $T_j(t)$.

Since  $\int_{\eta_\ell, \eta_j}g_{T_\ell(t)}(\eta_\ell)g_{T_j(t)}(\eta_j)d\eta_\ell d\eta_j =1,$ it follows that
\begin{align}
f(T_\ell(t),T_j(t)) &= 1 -  \int_{+\infty\geq\eta_j>\eta_\ell\geq 0}g_{T_\ell(t)}(\eta_\ell)g_{T_j(t)}(\eta_j)d\eta_\ell d\eta_j \nonumber \\
&= 1 - \int^{+\infty}_{\eta}g_{T_j(t)}(\eta_j) \big [1-Q_{1}(\sqrt{T_\ell(t)},\sqrt{\eta})\big ]d\eta_j.
\end{align}
Therefore, $f(T_\ell(t),T_j(t)) $ is monotonically increasing with respect to $T_\ell(t)$.

\section{Proof of Lemma~\ref{Lemma_infinite_sum}}\label{Appendix_proof_infinite_sum}
Following \eqref{Eq:finite_prob}, $f(x,y)$ can be represented as:
\begin{align}
f(x,y)        & = \int_{0}^{+\infty}g_{x}(\eta_1)\int_{0}^{\eta_1}g_{y}(\eta_2)d\eta_2 d\eta_1  =  1 - \int^{+\infty}_{0}g_{x}(\eta) Q_{1}\left(\sqrt{y},\sqrt{\eta}\right)d\eta, \label{eq:lemma3_2}
\end{align}
where we have used the fact that $\int_{0}^xg_{y}(\eta)d\eta = 1-Q_1(\sqrt{x},\sqrt{\eta})$ to obtain the second equation.

Substituting $g_{x}(\eta)$ of \eqref{eq:chi2_density} into \eqref{eq:lemma3_2}, it can be obtained that
\begin{align}
f(x,y) &= 1 - \int^{+\infty}_{0}\frac{1}{2}\exp(-\frac{x+\eta}{2})I_0(\sqrt{\eta x}) Q_{1}(\sqrt{y},\sqrt{\eta})d\eta \nonumber\\
& = 1 - \sum_{k=0}^{+\infty}\frac{\exp(-\frac{x}{2})x^k}{2\times4^k\times(k!)^2}\underbrace{\int^{+\infty}_{0}\eta^k\exp(-\frac{\eta}{2}) Q_{1}(\sqrt{y},\sqrt{\eta})d\eta}_{F(k,\eta,y)},\label{eq:lemma3_5}
\end{align}
where we have used the fact that the modified Bessel function of the first kind with zero-order can be represented as~\cite[Eq. 9.6.12 on Page 375]{abramowitz1965handbook}, $I_0(z) = \sum_{k=0}^{\infty}\frac{\left(\frac{1}{4}z^2\right)^k}{(k!)^2}$
to obtain \eqref{eq:lemma3_5}. To further compute $F(k,\eta,y)$ in \eqref{eq:lemma3_5}, we use the following fact~\cite[Eq. (15)]{cui2012two}:
\begin{align}
&R_{a,b}(M,p,1) = \frac{\Gamma(M)}{p^M}\times\left[ 1 - \frac{b^2}{b^2+2p}\exp\left(-\frac{a^2p}{b^2+2p}\right)\sum_{m=0}^{M-1}\left(\frac{2p}{b^2+2p}\right)^m L_m\left(-\frac{a^2p}{b^2+2p}\right)\right]\label{eq:lemma3_fromref}
\end{align}
where $R_{a,b}(k,p,v) \doteq \int_{0}^{+\infty}x^{k-1}\exp(-px)Q_v(a,b\sqrt{x})dx,$ and $\Gamma(M) = (M-1)!$ is the gamma function.
It is easy to see that $F(k,\eta,y) = R_{\sqrt{y},1}(k+1,\frac{1}{2},1)$, with $a = \sqrt{y}$, $b = 1$, $p = \frac{1}{2}$ and $M = k+1$ in \eqref{eq:lemma3_fromref}, thus $F(k,\eta,y)$ can be represented as:
\begin{align}
F(k,\eta,y)  & = \frac{\Gamma(k+1)}{(\frac{1}{2})^{k+1}}\left[ 1 - \frac{1}{2}\exp\left(-\frac{y}{4}\right)\sum_{m=0}^k\left(\frac{1}{2}\right)^mL_m\left(-\frac{y}{4}\right)\right ].\label{eq:lemma3_6}
\end{align}
Substituting \eqref{eq:lemma3_6} into \eqref{eq:lemma3_5}, it can be obtained that
\begin{align}
f(x,y) & = \frac{1}{2}\exp\left(-\frac{x}{2}-\frac{y}{4}\right) \underbrace{\sum_{k=0}^{+\infty} \frac{(\frac{x}{2})^k}{k!}\sum_{m=0}^{k}\frac{1}{2^m}L_m\left(-\frac{y}{4}\right)}_{G},\label{eq:lemma3_7}
\end{align}
where the $G$ term can be further computed as:
\begin{align}
G = & \exp\left(\frac{x}{2}\right)\sum_{m=0}^{+\infty}\left(\frac{1}{2}\right)^mL_m\left(-\frac{y}{4}\right)-\sum_{m=1}^{+\infty}\left(\frac{1}{2}\right)^mL_m\left(-\frac{y}{4}\right)\sum_{k=0}^{m-1}\frac{(\frac{x}{2})^k}{k!}, \\
=&  2\times\exp\left(\frac{x}{2} + \frac{ y}{4}\right) - \sum_{m=1}^{+\infty}\left(\frac{1}{2}\right)^mL_m\left(-\frac{y}{4}\right)\sum_{k=0}^{m-1}\frac{(\frac{x}{2})^k}{k!},\label{eq:lemma3_8}
\end{align}
since the Laguerre polynomial has the following property~\cite[Pg. 242]{magnus2013formulas}:
\begin{equation}
\sum_{m=0}^{+\infty}L_m(z)\omega^m = \frac{1}{1-\omega}\exp\left(\frac{\omega z}{\omega-1}\right), ~|\omega|<1.
\end{equation}
This completes the proof.

\section{Some analysis of Beam Deactivation \label{sec_accept}}
We now present some performance results of IDBS. As we aim to gain insights on the deactivation operation, we only consider one phase of search, e.g., Phase 1. Consider in this case the Tx omnidirectional transmit the pilot signals and the Rx scans $[0^\circ,360^\circ]$ using $M$ ideal beams, with the same beam widths, the same uniform gain with the intended coverage interval, and zero-leakage outside this interval~\cite{min2019TWC}.  Suppose also that the channel has a single path with some fixed angle $\phi$ at the Rx. Then there is only one true beam that has non-zero beamforming gain to the channel and this gain is $M$. Without loss of generality, we assume the true beam is beam~$1$. Consequently, $T_1(t)$ is a sequence of $\chi_2^2(\eta_1(t))$ random variables with
$\eta_1(t)  = {2 t P_T \abs{h_1}^2}/{\sigma^2}= {2 t P_T M |\gamma|^2}/{\sigma^2}$,
where $\gamma$ is the complex path coefficient (see \eqref{Eq:single_path}) and the pre-beamforming SNR is $ {P_T \abs{\gamma}^2}/{\sigma^2}$.  

\begin{lemma}
	Consider the event $E(t)$ that the true beam is first deactivated at the iteration $t$ under the algorithm. It holds that
	\begin{align}
	Pr\{E(t)\} & \leq Pr\{f(Y(t), T_1(t)) > \alpha \} = Pr\{ \tau_\alpha(Y(t)) > T_1(t)\} \doteq q(t),\label{eqn_RejectTrue}
	\end{align}
	where $Y(t)$ is the maximum of all the false beam measurements at iteration $t$ which are supposed taken, irrespective of whether they
	have been deactivated.
\end{lemma}
\begin{proof}
	Let $Z(t)$ be the maximum of the false beams active at step $t$ and $Y(t)$  the maximum of all false beams supposing they are all measured by $t$ times. Clearly $Y(t)\geq Z(t)$. Then $E(t)$ implies the event $ \lc f(Z(t), T_1(t)) > \alpha \rc$ which implies the event $\lc f(Y(t), T_1(t)) > \alpha \rc$.
\end{proof}

The bound $q(t)$ can be evaluated as $T_1(t)\sim \chi_2^2(\eta_1(t))$ and $Y(t)$, the maximum of $(M-1)$ $\chi_2^2$-random variables, are marginally independent. Fig.~\ref{fig_MaxUnionRej} graphs $q(t)$ for $\alpha = 0.90, 0.95, 0.97, 0.99$ and $M= 16$ with the pre-beamforming SNR $= -20$dB. We note that when producing these results, quadratic approximations to $\tau_\alpha(x)$ are used. It can be seen that $q(t)$ is decreasing exponentially with the number of iterations as the non-centrality parameter $\eta_1(t)$ increases. This means that if the true beam is ever deactivated, it more likely occurs in an early stage of the iterations, i.e., when $t$ is small.

It can also be seen that $q(t)$ is very small ($<10^{-3}$) even at the first iteration for $\alpha = 0.99$, thus the overall probability that the true beam is deactivated is small. This means that the true beam will be very likely in the active set throughout the iterations. This result also holds for smaller values of $\alpha$ provided the pre-beamforming SNR is sufficiently high. To see this, we present $Q^\alpha_{\infty} \doteq \sum_{t=1}^{\infty}q(t)$ the union bound of the overall probability of deactivation for four different $\alpha$ values and two different SNRs in Table \ref{table:UnionBnd}.  As can be seen from the table, when the pre-beamforming SNR $= -10$ dB, a low $\alpha = 0.9$ still has a small overall probability of deactivation as $Q^{\alpha}_{\infty} \approx  0.073$.

\begin{figure}[t]
	\centering
	\includegraphics[width=0.6\textwidth]{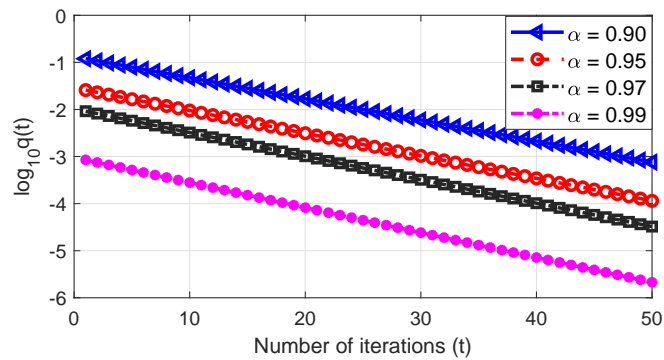}
	\caption{Upper bound on the probability that the true beam is rejected at iteration $t$: Ideal beam and single path channel assumed. Pre-beamforming SNR is $-20$ dB and $M=16$ beams.}\label{fig_MaxUnionRej}

\end{figure}

\begin{table}[t]
	\caption{Union Bound on Probability of Deactivation $Q^\alpha_{\infty}$. \label{table:UnionBnd}}
	\centering
	\small{
		\begin{tabular}{|c|c|c|c|c|}
			\hline
			$\alpha$	& 0.90 & 0.95 & 0.97 & 0.99 \\
			\hline
			$\gamma = -20$ dB & - &0.2443 & 0.0836 & 0.0073  \\
			\hline
			$\gamma = -10$ dB & 0.0734 &0.0142 & 0.0047 & 0.000397  \\
			\hline
		\end{tabular}
	}
\end{table}

To summarise our observations from evaluating $q(t)$, the true beam tends to be deactivated when $\alpha$ is small and the pre-beamforming SNR is low. If the true beam is ever deactivated, it tends to occur in an early iteration ($t$ is small).

\bibliographystyle{IEEETran}
\bibliography{IEEEabrv,Reference_mmWave}

\end{document}